\documentclass[12pt]{article}

\usepackage{geometry}
\usepackage{bbm}

\usepackage[utf8]{inputenc}

\usepackage{amssymb,amsmath,amsthm}

\usepackage{xcolor}
\usepackage{enumerate}

\usepackage[colorlinks=true,linkcolor=blue,citecolor=red]{hyperref}

\theoremstyle{plain}
\newtheorem{theorem}{Theorem}[section]
\newtheorem{lemma}[theorem]{Lemma}

\newtheorem{claim}[theorem]{Claim}
\newtheorem{corollary}[theorem]{Corollary}

\theoremstyle{definition} 

\newtheorem{remark}[theorem]{Remark}

\newcommand{\1}{\mathbbm{1}} 

\makeatletter
\def\@gifnextchar#1#2#3{\let\@tempe#1\def\@tempa{#2}\def\@tempb{#3}%
  \futurelet\@tempc\@gifnch}

\def\@gifnch{\ifx\@tempc\@sptoken\let\@tempd\@tempb%
  \else\ifx\@tempc\@tempe\let\@tempd\@tempa\else\let\@tempd\@tempb\fi\fi\@tempd}

\def\SK@set#1{\left\{#1\right\}}
\def\SK@@set#1#2{\{#1\,:\,
    \begin{array}{@{}l@{}}#2\end{array}
\}}
\def\SK@mset#1{\left\{\!\!\left\{#1\right\}\!\!\right\}}
\def\SK@@mset#1#2{\{\!\!\{#1\,:\,
    \begin{array}{@{}l@{}}#2\end{array}
\}\!\!\}}
\def\BIG@set#1{\Big\{#1\Big\}}
\def\BIG@@set#1#2{\Big\{#1\:\Big|\:
    \begin{array}{@{}l@{}}#2\end{array}
\Big\}}
\newcommand{\Set}[1]{\@gifnextchar\bgroup{\SK@@set{#1}}{\SK@set{#1}}}
\newcommand{\Mset}[1]{\@gifnextchar\bgroup{\SK@@mset{#1}}{\SK@mset{#1}}}
\newcommand{\Bigset}[1]{\@gifnextchar\bgroup{\BIG@@set{#1}}{\BIG@set{#1}}}
\makeatother

\newcommand{\refeq}[1]{(\ref{eq:#1})}
\newcommand{\setdef}[2]{\left\{ \hspace{0.5mm} #1 : \hspace{0.5mm} #2 \right\}}
\newcommand{\of}[1]{\left( #1 \right)}
\newcommand{\function}[2]{:#1 \rightarrow #2}

\newcommand{\bZ}{\mathbb{Z}}
\newcommand{\cP}{\mathcal{P}}

\newcommand{\creq}{\equiv_{\mathrm{CR}}}

\DeclareMathOperator{\aut}{Aut}
\DeclareMathOperator{\core}{core}
\newcommand{\compn}{\mathsf{H}_n}
\newcommand{\coren}{\mathsf{C}_n}
\newcommand{\corencol}{\mathsf{C}'_n}

\title{Canonical labeling of sparse random graphs}

\author{\Large Oleg Verbitsky\thanks{Institut f\"ur Informatik,
  Humboldt-Universit\"at zu Berlin, Unter den Linden 6, D-10099 Berlin.
  Supported by DFG grant KO 1053/8--2.
  On leave from the IAPMM, Lviv, Ukraine.}, \,\,\,\,\, Maksim Zhukovskii\thanks{The University of Sheffield, School of Computer Science, Sheffield S1 4DP, UK.\newline Email: m.zhukovskii@sheffield.ac.uk.}}

\date{}

\begin{document}

\maketitle

\begin{abstract}
  We show that if $p=O(1/n)$, then the Erd\H{o}s-R\'{e}nyi random graph $G(n,p)$ with high probability
  admits a canonical labeling computable in time $O(n\log n)$. Combined with the previous results on
  the canonization of random graphs, this implies that $G(n,p)$ with high probability
  admits a polynomial-time canonical labeling whatever the edge probability function~$p$.
  Our algorithm combines the standard color refinement routine with simple post-processing
  based on the classical linear-time tree canonization. Noteworthy, our analysis of how well
  color refinement performs in this setting allows us to complete the description of
  the automorphism group of the 2-core of~$G(n,p)$.
\end{abstract}

\section{Introduction}
\label{sc:intro}

On an $n$-vertex input graph $G$, a {\it canonical labeling algorithm} computes a bijection
$\lambda_G\function{V(G)}{\{1,\ldots,n\}}$ such that if another graph $G'$ is isomorphic to $G$,
then the isomorphic images of $G$ and $G'$ under respective permutations $\lambda_G$ and $\lambda_{G'}$ are equal.
Given the labelings $\lambda_G$ and $\lambda_{G'}$, it takes linear time to check whether $G$ and $G'$ are isomorphic.
The existence of polynomial-time algorithms for testing isomorphism of two given graphs and, in particular, for producing
a canonical labeling remain open. Babai's breakthrough quasi-polynomial algorithm for testing graph isomorphism \cite{Babai-GI}
was subsequently extended to a canonical labeling algorithm of the same time complexity \cite{Babai-CL}.
In the present paper, we address the canonical labeling problem for the Erd\H{o}s-R\'{e}nyi (or binomial) random graph
$G(n,p)$. Recall that the vertex set of $G(n,p)$ is $\{1,\ldots,n\}$, and each pair of vertices is adjacent with probability
$p=p(n)$, independently of the other pairs.

Babai, Erd\H{o}s, and Selkow~\cite{BES} proved that the simple algorithmic routine known as \emph{color refinement}
(\emph{CR} for brevity) with high probability produces a \emph{discrete} coloring of the vertices of $G(n,1/2)$, that is,
a coloring where the vertex colors are pairwise different. Since the vertex colors are isomorphism-invariant,
this yields a canonical labeling of $G(n,1/2)$ by numbering the color names in the lexicographic order.
Here and throughout, we say that an event happens for $G(n,p)$ \emph{with high probability}
(\emph{whp} for brevity) if the probability of this event tends to 1 as $n\to\infty$.
The result of \cite{BES} has a fundamental meaning: almost all graphs admit an easily computable canonical labeling
and, hence, the graph isomorphism problem has low average-case complexity.

The argument of \cite{BES} can be extended to show \cite[Theorem~3.17]{Bollobas_book} that the
CR coloring of $G(n,p)$ is whp discrete for all $n^{-1/5}\ln n\ll p\leq 1/2$.
Note that it is enough to consider the case of $p\leq 1/2$ since $G(n,1-p)$ has the same
distribution as the complement of $G(n,p)$. Remarkably,
the algorithm of Babai, Erd\H{o}s, and Selkow performs only a bounded number of color refinement steps
and, due to this, works in linear time.

A different algorithm suggested by Bollob\'{a}s~\cite{Bol} works in polynomial time and whp produces a canonical labeling
of $G(n,p)$ in a much sparser regime, namely when $c_1\frac{\ln n}{n}\leq p\leq c_2\, n^{-11/12}$ for some positive
constants $c_1$ and $c_2$. The next improvement was obtained by Czajka and Pandurangan~\cite{CzP} who proved that,
in a bounded number of rounds, CR yields a discrete coloring of $G(n,p)$ whp
when $\frac{\ln^4 n}{n\ln\ln n}\ll p\leq \frac{1}{2}$. Finally, Linial and Mosheiff~\cite{LM}
designed a polynomial time algorithm for canonical labeling of $G(n,p)$ when $\frac{1}{n}\ll p\leq\frac{1}{2}$.
As shown by  Gaudio, R\'{a}cz, and Sridhar \cite{GaudioRS23},
in the subdiapason $p\ge\frac{(1+\delta)\ln n}{n}$ for any fixed $\delta>0$, whp
a canonical labeling can still be provided by CR in a bounded number of rounds.

The decades-long line of research summarized above leaves open the question whether a random
graph $G(n,p)$ admits efficient canonization in the regime $p=O(1/n)$.
Note that the case of $p=o(1/n)$ is easy. Indeed, as long as $pn=1-\omega(n^{-1/3})$,
whp $G(n,p)$ is a vertex-disjoint union of trees and \emph{unicyclic} graphs (i.e., connected graphs containing exactly one cycle)~\cite[Theorem 5.5]{Janson_book}.
Canonization of such graphs is tractable due to the classical linear-time canonical labeling
algorithms for trees \cite{AHU} and even planar graphs (see \cite{Babai81}
for a survey of the early work on graph isomorphism covering these graph classes).
Thus, efficient canonization remains unknown for all $p=p(n)$ such that, for some $C>0$ and all $n$, $1-Cn^{-1/3}\leq pn\leq C$ (even though $G(n,p)$ stays planar with a non-negligible probability as long as $pn=1+O(n^{-1/3})$; see \cite{NoyRR15}).
 Our first result closes this gap, implying that the Erd\H{o}s-R\'{e}nyi random graph $G(n,p)$
 whp admits an efficiently computable canonical labeling, whatever the edge probability function~$p(n)$
 (see a formal proof of this implication in Appendix~\ref{sc:appendix}).

\begin{theorem}\label{thm:CanLab}
  If $p=O(1/n)$, then $G(n,p)$ whp admits a canonical labeling computable in time~$O(n\log n)$.
\end{theorem}

We now recall some highlights of the evolution of the random graph. Erd\H{o}s and R\'{e}nyi~\cite{ER-evolution}
proved their spectacular result that when $p$ passes a certain threshold around $1/n$,
then the size of the largest connected component in $G(n,p)$ rapidly grows from $\Theta(\log n)$ to $\Theta(n)$.
A systematic study of the structure of connected components in the random graph when $p$ is around the critical
value $1/n$ was initiated in the influential paper of Bollob\'{a}s~\cite{Bol-evolution}.
For more details about the phase transition, see, e.g.,~\cite[Chapter~5]{Janson_book}.

A connected graph is called {\it complex}, if it has more than one cycle.
The union of all complex components of a graph $G$ will be called the \emph{complex part} of $G$,
and the union of the other components will be referred to as the \emph{simple part}.
As we already mentioned, if $pn=1-\omega(n^{-1/3})$ then the complex part of $G(n,p)$ is whp empty.
This is the so-called \emph{subcritical phase}.
In the {\it critical phase}, when $pn=1\pm O(n^{-1/3})$,
the complex part of $G(n,p)$ whp has size $O_P(n^{2/3})$ and its structure is thoroughly described in~\cite{Luczak-auto,LPW}.
Here and below, for a sequence of random variables $\xi_n$ and a sequence of reals $a_n$
we write $\xi_n=O_P(a_n)$ if the sequence $\xi_n/a_n$ is stochastically bounded%
\footnote{I.e. for every $\varepsilon>0$, there exists $C>0$ and $n_0$ such that $\mathbb{P}(|\xi_n/a_n|>C)<\varepsilon$ for all $n\geq n_0$.}.
Finally, in the {\it supercritical phase}, when $pn=1+\omega(n^{-1/3})$, whp $G(n,p)$ contains a unique complex connected component and this component has size $\Theta(n^2(p-1/n))$.
 In particular, when $p=O(1/n)$ and $p>(1+\delta)/n$ for a constant $\delta>0$
(we refer to this case as \emph{strictly supercritical regime}), whp this component has linear size $\Omega(n)$.
It is called the \emph{giant component} as all the other connected components have size~$O(\log n)$.

In general, the simple part of a graph $G$ is easily canonizable by the known techniques,
which reduces our problem to finding a canonical labeling for the complex part of $G$.
Furthermore, recall that the 2-core of a graph $H$, which we will for brevity call just \emph{core}
and denote by $\core(H)$, is the maximal subgraph of $H$ that does not have vertices of degree 1.
Equivalently, $\core(H)$ can be defined as the subgraph of $H$ obtained by iteratively pruning all vertices
in $H$ that have degree at most 1 until there are no more such vertices. Thus, if $H$ is the (non-empty)
complex part of $G$, then $H$ consists of $\core(H)$  and some (possibly empty) rooted trees planted
at the vertices of the core. It follows that if we manage to canonically label the vertices of $\core(H)$,
then this labeling easily extends to a canonical labeling of the entire graph~$G$.

Suppose that CR is run on $H$. In the most favorable case, it would output a vertex coloring discrete on $\core(H)$.
It turns out that, though not exactly true, this is indeed the case to a very large extent.

\begin{theorem}\label{thm:ColRef}
  Let $G_n=G(n,p)$ and assume that $p=O(1/n)$. Let $\compn$ denote the complex part of $G_n$ and $\coren=\core(\compn)$.
  When CR is run on $\compn$, then
  \begin{enumerate}[\bf 1.]
  \item
    CR assigns individual colors to all but $O_P(1)$ vertices in $\coren$;
  \item
    the other color classes whp have size 2;
  \item
    whp, every such color class is an orbit of the automorphism group $\aut(\compn)$
    consisting of two vertices with degree 2 in~$\coren$.
  \end{enumerate}
\end{theorem}

\begin{remark}
From our proofs it is easy to derive that, when $np=1+o(1)$,
whp CR distinguishes between all vertices of $\coren$ whp.
\label{rk:action}
\end{remark}

Theorem \ref{thm:ColRef} allows us to obtain an efficient canonical labeling algorithm for $G(n,p)$,
as stated in Theorem \ref{thm:CanLab}, by combining CR with simple post-processing
whose most essential part is invoking the linear-time tree canonization.
Another consequence of Theorem \ref{thm:ColRef} is that CR alone is powerful enough
to solve the standard version of the graph isomorphism problem for the complex part of $G(n,p)$.
Specifically, we say that a graph $H$ is \emph{identifiable} by CR if CR distinguishes $H$ from any non-isomorphic
graph $H'$ (in the sense that CR outputs different multisets of vertex colors on inputs $H$ and~$H'$).
It is not hard to see that $H$ is identifiable by CR whenever the CR coloring of $H$ is discrete.
Fortunately, the properties of the CR coloring ensured by Theorem \ref{thm:ColRef} are still sufficient
for CR-identifiability.

\begin{corollary}\label{cor:ident}
  Under the assumption of Theorem \ref{thm:ColRef},
  \begin{enumerate}[\bf 1.]
  \item
    $\compn$ is whp identifiable by CR and, consequently,
  \item
whp, $G_n$ is identifiable by CR exactly when the simple part of $G_n$ is identifiable.
  \end{enumerate}
\end{corollary}
The CR-identifiability of the simple part of a graph admits an explicit, 
efficiently verifiable characterization, which we give in Theorem~\ref{thm:amenability}. This characterization can be used to show that the random graph $G_n$ is identifiable by CR with probability asymptotically bounded away from 0 and 1.

Our techniques for proving Theorems \ref{thm:CanLab} and \ref{thm:ColRef} can also be used
for deriving a structural information about the automorphisms of a random graph.
As proved by Erd\H{o}s and R\'{e}nyi~\cite{ER-asymmetry} and by Wright~\cite{Wright},
$G(n,p)$ for $p\leq 1/2$ is asymmetric, i.e., has no non-identity automorphism, if
$np-\ln n\to\infty$ as $n\to\infty$. This result is best possible because if, $np-\ln n\leq \gamma$ for some constant $\gamma>0$,
then the random graph has at least 2 isolated vertices with non-vanishing probability.
It is noteworthy that the asymmetry of $G(n,p)$ in the regime $p\ge\frac{(1+\delta)\ln n}{n}$ can be
certified by the fact that CR coloring of $G(n,p)$ is discrete due to the
aforementioned result of Gaudio, R\'{a}cz, and Sridhar \cite{GaudioRS23}.
In the diapason of $p$ forcing $G(n,p)$ to be disconnected, the action of the automorphism group
can be easily understood on the simple part and on the tree-like pieces of the complex part,
and full attention should actually be given to the core of the complex part.
Theorem \ref{thm:ColRef} provides a pretty much precise information about the action
of $\aut(G_n)$ on $\coren$.
More subtle questions arise if, instead of considering the action of $\aut(G_n)$ on $\coren$,
we want to understand the automorphisms of $\coren$ itself. It is quite remarkable that
the CR algorithm, applied to $\coren$ rather than to $\compn$, can serve as a sharp instrument
for tackling this problem (and, in fact, the proof of Theorem \ref{thm:ColRef} is based on an analysis
of the output of CR on~$\coren$).

We begin with describing simple types of potential automorphisms of $\coren$
(with the intention of showing that, whp, all automorphism of $\coren$ are actually of this kind).
If a vertex $x$ has degree 2 in $\coren$, then it belongs to a (unique) path
from a vertex $s$ of degree at least 3 to a vertex $t$ of degree at least 3 with all
intermediate vertices having degree 2. We call such a path in $\coren$ \emph{pendant}.
It is possible that $s=t$, and in this case we speak of a \emph{pendant cycle}.
Clearly, the reflection of a pendant cycle fixing its unique vertex of degree more than 2
is an automorphism of $\coren$. Furthermore, we call two pendant paths \emph{transposable}
if they have the same length and share the endvertices. Note that $\coren$ has an automorphism
transposing such paths (and fixing their endvertices). Let $A_1$ denote the set of the automorphisms
provided by pendant cycles, and let $A_2$ be the set of the automorphisms provided by transposable
pairs of pendant paths. Moreover, $\coren$ can have a connected component consisting of two vertices
of degree 3 and three pendant paths of pairwise different lengths between these vertices.
Such a component has a single non-trivial automorphism, which contributes in $\aut(\coren)$. 
The set of such automorphisms of $\coren$ will be denoted by~$A_3$. 

Recall that an elementary abelian 2-group is a group in which all non-identity elements
have order 2 or, equivalently, a group isomorphic to the group $(\bZ_2)^k$ for some~$k$.
 
\begin{theorem}
Let $G_n=G(n,p)$ and assume that $p=O(1/n)$. Let $\coren$ be the core of the complex part of~$G_n$.
\begin{enumerate}[\bf 1.]
 \item
  The order of $\aut(\coren)$ is stochastically bounded, i.e., $|\aut(\coren)|=O_P(1)$. 
\item
  Whp, $\aut(\coren)$ is an elementary abelian 2-group. Moreover,
  $A_1\cup A_2\cup A_3$ is a minimum generating set of $\aut(\coren)$.\footnote{Consequently, whp $\coren$
    contains neither a triple of pairwise transposable paths, nor two isomorphic components with an automorphism in $A_3$,
    nor a cyclic component with a single chord between diametrically opposite vertices. Moreover, whp no two pendant cycles
    in $\coren$ share a vertex.}
\item
 In addition,
\begin{enumerate}[\bf (a)]
\item
  if $pn\geq 1+\delta$ for a constant $\delta>0$,
  then  both $A_1$ and $A_2$ are non-empty with non-negligible probability, while $A_3=\varnothing$ whp.
\item
  If $pn=1+o(1)$ and $pn=1+\omega(n^{-1/3})$,
  then $A_1\neq\varnothing$ with non-negligible probability, while $A_2=A_3=\varnothing$ whp.
\item
  If $pn=1\pm O(n^{-1/3})$, then both $A_1$ and $A_3$ are non-empty with non-negligible probability, while $A_2=\varnothing$ whp.
\end{enumerate}
\end{enumerate}
\label{th:symmetries}
\end{theorem}

This theorem makes a final step in the study of the automorphisms group of a random graph.
Recall that $\compn$ is whp empty when $np=1-\omega(n^{-1/3})$ and that $G(n,p)$ is connected and asymmetric
when $np=\ln n+\omega(1)$. We, therefore, focus on the intermediate diapason.
If $np\to\infty$ as $n\to\infty$, then the core of the giant component of $G(n,p)$ is whp still asymmetric,
as proved independently by \L uczak~\cite{Luczak-symmetries} and Linial and Mosheiff~\cite{LM}.
Moreover, \L uczak described the automorphisms group of the core of the giant component of $G(n,p)$
when $np>\gamma$ for a large enough constant $\gamma$, and obtained an analogue of Theorem~\ref{th:symmetries} for this case;
see~\cite[Theorem 4]{Luczak-symmetries}. Our Theorem~\ref{th:symmetries} not only extends \cite[Theorem 4]{Luczak-symmetries}
to the full range of $p=O(1/n)$ but also refines this result even for $np>\gamma$
by showing that $\aut(\coren)$ is actually an elementary 2-group. Another interesting observation is that some automorphisms of the core do not extend to automorphisms of the entire $G(n,p)$. Indeed, if $np=1+o(1)$, then whp $\aut(G(n,p))$ acts trivially on the core; see Remark~\ref{rk:action}.

\begin{remark}
Linial and Mosheiff~\cite[Section 6]{LM} observed that the core of the giant component of $G(n,p)$
in the strictly supercritical regime 
 has non-trivial automorphisms with a non-vanishing probability. They also expressed a suspicion that the unique 2-connected
component of the giant of size $\Theta(n)$, whose existence is known due to Pittel~\cite{Pittel}, is asymmetric whp.
However, in the strictly supercritical regime, the automorphisms described in Theorem~\ref{th:symmetries} force even
the large 2-connected component
 to have non-trivial automorphisms with a non-vanishing probability, ruling this possibility out.
It is worth noting that such automorphisms whp disappear when $np=1+o(1)$.
\end{remark}

\paragraph{Related work.}
As we already mentioned, Theorem \ref{thm:CanLab} combined with the previous results on canonical labeling of $G(n,p)$ for $1/n\ll p\le 1/2$ implies the existence of a polynomial-time
canonical labeling algorithm succeeding on $G(n,p)$ whp for an \emph{arbitrary} edge probability
function $p=p(n)$.
 In this form, this result has been independently obtained by
Michael Anastos, Matthew Kwan, and Benjamin Moore~\cite{AnastosKM24}.
Another result in their paper describes the action of $\aut(G(n,p))$ on the core of $G(n,p)$,
which follows also from our Theorem~\ref{thm:ColRef} and the results of \L uczak~\cite{Luczak-symmetries} and Linial and Mosheiff~\cite{LM}.
The techniques used in \cite{AnastosKM24} and in our paper are completely different, though both proofs rely on color refinement.
The two approaches have their own advantages. The method developed in \cite{AnastosKM24}
is used there also to show that color refinement is helpful for canonical labeling of the random graph when $p\gg 1/n$ and to study the smoothed complexity of graph isomorphism.
Our method allows obtaining precise results on the automorphism group of the core
(Theorem~\ref{th:symmetries}).

Immerman and Lander \cite{Immerman1990} showed a tight connection between CR-identifiability
and definability of a graph in first-order logic with counting quantifiers. Corollary \ref{cor:ident}
can, therefore, be recast in logical terms as follows. If $p=O(1/n)$, then $\compn$ is whp definable
in this logic with using only two first-order variables (where the definability of a graph $H$ means the existence of
a formula which is true on $H$ and false on any graph non-isomorphic to $H$). Definability of
the giant component of $G(n,p)$ in the standard first-order logic (without counting quantifiers)
was studied by Bohman et al.~\cite{BohmanFLPSSV07}.
A closer analysis of the argument in \cite{BohmanFLPSSV07} reveals that it can be used
to show that the 2-dimensional Weisfeiler-Leman algorithm (a more powerful version of
CR corresponding to the counting logic with 3 variables) individualizes the vertices of
degree at least 3 in the core of the giant component (in the strictly supercritical
regime). Our Theorem \ref{thm:ColRef} yields a much stronger version of this fact.

\paragraph{Structure of the paper and proof strategy.}
In Section \ref{sc:CR}, we derive Theorem \ref{thm:CanLab} (as well as Corollary \ref{cor:ident})
from Theorem \ref{thm:ColRef}, while the proofs of Theorems \ref{thm:ColRef} and \ref{th:symmetries}
are postponed to Section \ref{sc:proofs_th_Intro}.
Section \ref{sc:CR} begins with formal description of the color refinement algorithm
in Subsection \ref{ss:CR-descr} and then, in Subsection \ref{sc:universal_cover}, presents
a useful criterion of CR-distinguishability in terms of universal covers.
The concept of a universal cover appeared in algebraic and topological graph theory~\cite{Biggs,CDS,Massey},
and its tight connection to CR was observed in~\cite{Angluin}. Subsection \ref{sc:CR_unicyclic}
pays special attention to the CR-identifiability of unicyclic graphs, which in Subsection \ref{ss:criterium}
allows us to obtain an explicit criterion of CR-identifiability for general graphs in terms of the
complex and the simple part of a graph.

The main tool for proving Theorems \ref{thm:ColRef} and \ref{th:symmetries}
is Lemma~\ref{th:giant_main} (our Main Lemma), whose proof occupies Section \ref{cs:rg_amenability}.
This lemma says that CR is unable to distinguish between two vertices in the core
only if they lie either on pendant paths (with the same endvertices) transposable by an automorphism of the graph
or on a pendant cycle admitting a reflection by an automorphism.
The proof exploits useful properties of the CR-coloring of the core collected in Subsection \ref{ss:CR-cores},
which are proved there with actively using the relationship between CR and universal covers.

The proof of Main Lemma also heavily relies on the notion of a kernel.
The {\it kernel} $K(G)$ of a graph $G$ is a multigraph obtained from $\core(G)$ by contracting all pendant paths.
That is, $K(G)$ is obtained via the following iterative procedure: at every step if there exists a vertex $u$
with only two neighbors $v_1,v_2$, we remove $u$ with both incident edges and add the edge $\{v_1,v_2\}$ instead. Note that this transformation can lead to appearance of multiple edges and loops.
 
To prove that CR colors vertices of the core in the described manner, we use several standard techniques
for analyzing the distributions of the core and the kernel, in particular, the contiguous models due to
Ding, Lubetzky, and Peres~\cite{DLP_anatomy} in strictly supercritical regime and due to
Ding, Kim, Lubetzky, and Peres in critical regime~\cite{DKLP:anatomy}. We recall definitions of the models
in Section~\ref{sc:complex_structure}. They allow to transfer properties of random multigraphs with given
degree sequences to the kernel of the giant component in the random graph. Another important ingredient in
our proofs is the fact that in the kernel of the supercritical random graph there are whp no small complex subgraphs.
In Section~\ref{sc:complex_structure}, we derive it from the contiguous models using properties of random
multigraphs with specific degree sequences (cf., e.g.,~\cite[Corollary 5]{Gao}).

In order to prove that CR distinguishes between all vertices in the kernel (or most vertices in the core),
we consider separately two types of vertices: first, we prove that CR colors differently all vertices such
that their large enough neighborhoods induce trees. This is done in Sections \ref{sc:proof2_large} and
\ref{sc:proof2_small} for $p=1+\omega(n^{-1/3})$ and $p=1+O(n^{-1/3})$ respectively. Then, in Section~\ref{sc:proof3},
we prove that these vertices are helpful to distinguish between all the remaining vertices. 

In Section~\ref{sc:main_critical}, we extend Main Lemma to the critical regime and describe
the outcome of CR on the core of the complex part. Its behavior is different from
supercritical regime due to complex components with only two cycles that appear in the random graph with
non-vanishing probability (cf.\ Theorem~\ref{th:symmetries} Part 3(c)).
The corresponding statement is Lemma~\ref{lm:critical_CR_core}.

Theorems~\ref{thm:ColRef}~and~\ref{th:symmetries} are derived from Main Lemma and
Lemma~\ref{lm:critical_CR_core}, as already said, in Section~\ref{sc:proofs_th_Intro}.

\section{Color refinement: From identifiability to canonical labeling}
\label{sc:CR}

\subsection{Description of the CR algorithm}\label{ss:CR-descr}

We now give a formal description of the \emph{color refinement} algorithm (\emph{CR} for short). CR operates on vertex-colored graphs but applies also to uncolored graphs by assuming that their vertices are colored uniformly. An input to the algorithm consists either of a single graph or a pair of graphs. Consider the former case first. For an input graph $G$ with initial coloring $C_0$, CR iteratively computes new colorings
\begin{equation}
  \label{eq:refinement}
C_{i}(x)=\of{C_{i-1}(x),\Mset{C_{i-1}(y)}_{y\in N(x)}},  
\end{equation}
where $\Mset{}$ denotes a multiset and $N(x)$ is the neighborhood of a vertex $x$.
Denote the partition of $V(G)$ into the color classes of $C_i$ by $\cP_i$.
Note that each subsequent partition $\cP_{i+1}$ is either finer than or equal to $\cP_i$.
If $\cP_{i+1}=\cP_i$, then $\cP_{j}=\cP_i$ for all $j\ge i$.
Suppose that the color partition stabilizes in the $t$-th round,
that is, $t$ is the minimum number such that $\cP_t=\cP_{t-1}$.
CR terminates at this point and outputs the coloring $C=C_t$.
Note that if the colors are computed exactly as defined by \refeq{refinement},
they will require exponentially long color names. To prevent this,
the algorithm renames the colors after each refinement step, using the same set
of no more than $n$ color names.

If an input consists of two graphs $G$ and $H$, then it is convenient to
assume that their vertex sets $V(G)$ and $V(H)$ are disjoint.
The vertex colorings of $G$ and $H$ define an initial coloring $C_0$ of
the union $V(G)\cup V(H)$, which is iteratively refined according to \refeq{refinement}.
The color partition $\cP_i$ is defined exactly as above but now on the whole
set $V(G)\cup V(H)$. As soon as the color partition of $V(G)\cup V(H)$ stabilizes,
CR terminates and outputs the current coloring $C=C_t$ of $V(G)\cup V(H)$.
The color names are renamed for both graphs synchronously.

We say that CR \emph{distinguishes} $G$ and $H$ if $\Mset{C(x)}_{x\in V(G)}\ne\Mset{C(x)}_{x\in V(H)}$. If CR fails to distinguish $G$ and $H$, then we call these graphs \emph{CR-equivalent} and write $G\creq H$. A graph $G$ is called \emph{CR-identifiable} if $G\creq H$ always implies $G\cong H$.

\subsection{Covering maps and universal covers}
\label{sc:universal_cover}

A surjective homomorphism from a graph $K$ onto a graph
$G$ is a \emph{covering map} if its restriction to the neighborhood of
each vertex in $K$ is bijective. 
We suppose that $G$ is a finite graph, while $K$ can be an infinite graph.
If there is a covering map from $K$ to $G$ (in other terms, $K$ \emph{covers} $G$), then $K$ is called a
\emph{covering graph} of $G$.  Let $G$ be connected. We say that a graph $U$ is a \emph{universal cover} of a graph
$G$ if $U$ covers every connected covering graph of $G$. A universal cover
$U=U^G$ of $G$ is unique up to isomorphism. Alternatively, $U^G$ can
be defined as a tree that covers $G$.  If $G$ is itself a tree, then
$U^G\cong G$; otherwise the tree $U^G$ is infinite.

A straightforward inductive argument shows that a covering map $\alpha$
preserves the coloring produced by CR, that is, $C_i(u)=C_i(\alpha(u))$
for all $i$, where $C_i$ is defined by \refeq{refinement}. 
It follows that, if two connected graphs $G$ and $H$ have a common universal cover,
i.e., $U^G\cong U^H$, then
$\setdef{C(u)}{u\in V(G)} = \setdef{C(v)}{v\in V(H)}$.
The converse implication is also true, as a consequence of the following lemma.

\begin{lemma}[cf.~Lemmas 2.3 and 2.4 in \cite{KrebsV15}]\label{lem:UvsC}
 Let $U^G$ and $U^H$ be universal covers of connected graphs $G$ and $H$ respectively. Furthermore, let $\alpha$ be a covering map from $U^G$ to $G$ and $\beta$ be a covering map from $U^H$ to $H$. For a vertex $x$ of $U^G$ and a vertex $y$ of $U^H$, let
$U^G_x$ and $U^H_y$ be the rooted versions of $U^G$ and $U^H$ with roots at $x$ and $y$ respectively.
Then $U^G_x\cong U^H_y$ (isomorphism of rooted trees) if and only if $C(\alpha(x))=C(\beta(y))$.
\end{lemma}

The union of CR-identifiable graphs does not need be CR-identifiable.
However, the concept of a universal cover allows us to
state the following criterion, which is an extension of \cite[Thm.~5.4]{ArvindKRV17}
(see \cite[p.~649]{ArvindKRV17} for details).

\begin{lemma}\label{lem:union-amenable}
  Let $G_1,\ldots,G_k$ be connected CR-identifiable graphs and $G$ be
  their vertex-disjoint union. Then $G$ is CR-identifiable if and only if,
  for every pair of distinct $i$ and $j$ such that neither $G_i$ nor $G_j$ is a tree, the universal covers of $G_i$ and $G_j$ are non-isomorphic.
\end{lemma}

\subsection{Unicyclic graphs}
\label{sc:CR_unicyclic}

\subsubsection{Universal covers of unicyclic graphs}\label{ss:UC-unicyclic}

For a unicyclic graph $G$, its $\core(G)$ is the set of vertices
lying on the unique cycle of $G$. We use the notation $c(G)=|\core(G)|$ for the length of this cycle.
For a vertex $x$ in $\core(G)$, let $G_x$ denote the subgraph of $G$ induced
by the vertices reachable from $x$ along a path avoiding the other vertices in $\core(G)$.
This is obviously a tree. Moreover, we define $G_x$ as a rooted tree with root at $x$.
Let $t(x)$ denote the isomorphism class of the rooted tree $G_x$.
We treat $t$ as a coloring of $\core(G)$ and write $R(G)$ to denote the cycle of $G$
endowed with this coloring. Thus, $R(G)$ is defined as a vertex-colored cycle graph.
It will also be useful to see $R(G)$ as a \emph{circular word} over the alphabet
$\Set{t(x)}{x\in \core(G)}$; see, e.g., \cite{HegedusN16} and the references therein 
for more details on this concept in combinatorics on words.
In fact, $R(G)$ is associated with two circular words, depending on one of the two directions
in which we go along $R(G)$. However, the choice of one of the two words is immaterial
in what follows.

Speaking about a \emph{word}, we mean a standard, non-circular word.
Two words are conjugated if they are obtainable from
one another by cyclic shifts. A circular word is formally defined as
the conjugacy class of a word. A word $u$ is a period of a word $v$
if $v=u^k$ for some $k\ge1$. A word $u$ is a period of a circular word $w$
if $u$ is a period of some representative in the conjugacy class of $w$.
Note that if $u$ is a period of a word $v$, then any conjugate of $u$
is a period of some conjugate of $v$. This allows us to consider
periods of circular words themselves being circular words.
We define the \emph{periodicity} $p(w)$ of a circular word $w$
to be the minimum length of a period of $w$. It may be useful to keep in mind
that a period of length $p(w)$ is also a period of every period of $w$;
cf.~\cite[Proposition 1]{HegedusN16} (note that our terminology is different from~\cite{HegedusN16}).

For a unicyclic graph $G$, we define its periodicity by $p(G)=p(R(G))$,
where $R(G)$ is seen as a circular word as explained above. Note that
$p(G)$ is a divisor of $c(G)$ and that
$1 \le |\Set{t(x)}_{x\in \core(G)}| \le p(G) \le c(G)$.

Like trees, unicyclic graphs are also characterizable in terms of universal covers.

\begin{claim}\label{lem:unique-path}
  A connected graph $G$ is unicyclic if and only if $U^G$ has a unique infinite path.
\end{claim}

The unique infinite path subgraph of $U^G$ will be denoted by $P(U_G)$.
The structure of $U^G$ is clear: The cycle of $G$ is unfolded into the infinite path $P(U^G)$.
Moreover, let $\alpha$ be a covering map from $U^G$ to $G$. Then $U^G$ is obtained
by planting a copy of the rooted tree $G_{\alpha(x)}$ at each vertex $x$ on $P(U^G)$.
The path $P(U^G)$ will be considered being a vertex colored graph, with each vertex $x$
colored by $t(\alpha(x))$.

The following observation is quite useful in what follows.
Let $\alpha$ be a covering map from $U^G$ to $G$. The restriction of $\alpha$ to $P(U^G)$
is a covering map from the vertex-colored path $P(U^G)$ to the vertex-colored cycle $R(G)$.
Note that a covering map must preserve vertex colors. 

\begin{lemma}\label{cor:RGRH}
  Let $G$ and $H$ be connected unicyclic graphs. Then $U^G\cong U^H$ if and only if
  the circular words $R(G)$ and $R(H)$ have a common period. Moreover, if $U^G\cong U^H$,
  then $p(G)=p(H)$.
\end{lemma}

\begin{proof}
In one direction the statement is clear: if $R(G)$ and $R(H)$ have a common period, then $U^G\cong U^H$ by the definition. Let $U\cong U^G\cong U^H$ be a common universal cover of $G$ and $H$.
  We can naturally see $P(U)$ as an infinite word.
  An arbitrary subword of length $p(G)$ of $P(U)$ is a period of $P(U)$,
  and the same is true for an arbitrary subword of length $p(H)$ of $P(U)$.
  It follows 
  that $P(U)$ has a period $u=u_1\ldots u_q$ of length $q=\gcd(p(G),p(H))$. Indeed, it is sufficient to note that there exist integers $\beta_1,\beta_2$ such that $q=\beta_1p(G)-\beta_2 p(H)$. Thus, for any $u_0,u_q$ at distance $q$ in $P(U)$, we get that $u_0=u_{\beta_ 1 p(G)}=u_{q+\beta_2 p(H)}=u_q$.

  If $\alpha$ and $\beta$ are covering maps from $U$ to $G$ and $H$ respectively,
  then $\alpha(u_1)\ldots \alpha(u_q)$ is a period of $R(G)$ and
  $\beta(u_1)\ldots \beta(u_q)$ is a period of $R(H)$.
  Since $\alpha$ and $\beta$ preserve the vertex colors, we have the equality
  $\alpha(u_1)\ldots \alpha(u_q)=\beta(u_1)\ldots \beta(u_q)$.
  This also implies that, in fact, $q=p(G)=p(H)$.
\end{proof}

\subsubsection{CR-identifiability of unicyclic graphs}
\label{sc:amenability_ALL}

\begin{claim}\label{lem:GHi}
  Let $G$ be a connected unicyclic graph. Suppose that $G\creq H$ and $H$ consists of
  connected components $H_1,\ldots,H_m$. Then
  \begin{enumerate}[\bf 1.]
  \item
    $U^{H_i}\cong U^G$ for all $i$,
  \item
    every $H_i$ is unicyclic, and
  \item
    $c(G)=c(H_1)+\cdots+c(H_m)$.
  \end{enumerate}
\end{claim}

\begin{proof}
  1. Fix $i\in[m]$.  Let $U^G$ and $W=U^{H_i}$ and $\alpha^U,\alpha^W$ be the respective covering maps.
  Let $y$ be a vertex in $U^{H_i}$. Since $G$ and $H$ are CR-equivalent,  $U^G$ must contain a vertex $x$ such that $C(\alpha^U(x))=C(\alpha^W(y))$. By Lemma \ref{lem:UvsC},
  $U_x\cong W_y$. It follows that $U\cong W$.

  2. Immediately by Part 1 and Claim \ref{lem:unique-path}.

  3. Fix a period of $R(G)$ and set $T=\sum_x |V(G_x)|$ where the summation goes over all $x$ in this period (this definition obviously does not depend on the choice of the period).
  Let $p=p(G)$.  Note that $|V(G)|=\frac{c(G)}{p}\,T$. By Part 1 and Lemma \ref{cor:RGRH}, we similarly have
  $|V(H_i)|=\frac{c(H_i)}{p}\,T$. The required equality now follows from the trivial  equality $|V(G)|=|V(H_1)|+\cdots+|V(H_m)|$.
\end{proof}

\begin{lemma}\label{thm:unicyclic}
  A connected unicyclic graph $G$ is CR-identifiable if and only if one of the following conditions is true:
  \begin{itemize}
  \item
    $p(G)=1$ and $c(G)\in\{3,4,5\}$,
  \item
    $p(G)=2$ and $c(G)\in\{4,6\}$,
  \item
    $p(G)=c(G)$.    
  \end{itemize}
\end{lemma}

\begin{proof}

( $\Leftarrow$ )
Suppose that a connected unicyclic graph $G$ satisfies one of the three conditions
and show that it is CR-identifiable. Assuming that $H$ is CR-equivalent to $G$, we have
to check that $G$ and $H$ are actually isomorphic.

Assume first that $H$ is connected. If $H$ is not unicyclic, then $G$ and $H$
have equal number of vertices but different number of edges. This implies that
$G$ and $H$ have different degree sequences, contradicting the assumptions that $G\creq H$.
Therefore, $H$ must be unicyclic. By Part 3 of Claim~\ref{lem:GHi}, we have $c(G)=c(H)$.
Along with Lemma \ref{cor:RGRH}, which is applicable because $U^G\cong U^H$
whenever $G\creq H$, this implies that $R(G)\cong R(H)$. The last relation, in its turn,
implies that $G\cong H$.

Assume now that $H$ is disconnected. Let $H_1,\ldots,H_m$ be the connected components of $H$.
Combining Claim~\ref{lem:GHi} and Lemma \ref{cor:RGRH}, we see that
$p(G)=p(H_1)\le c(H_1)<c(G)$. It follows that $G$ satisfies one of the first two
conditions. The restrictions on $c(G)$, however, rule out the equality in Part 3 of Claim~\ref{lem:GHi}. In particular, if $c(G)=6$, then the only possible case is $m=2$ and $c(H_1)=c(H_2)=3$. However, it contradicts the equality $p(H_1)=p(H_2)=p(G)=2$.
Thus, the case of disconnected $H$ is actually impossible, that is, all such $H$
are distinguishable from $G$ by~CR.

\medskip

\noindent
( $\Rightarrow$ )
Suppose that all three conditions are false. That is,
either $p(G)=1$ and $c(G)\ge6$, or $p(G)=2$ and $c(G)\ge8$ (note that $c(G)$ is even in this case),
or $3\le p(G)<c(G)$ (in the last case, $p(G)$ is a proper divisor of $c(G)$).
In each case, $R(G)$ is CR-equivalent to a disjoint union of two shorter vertex-colored
cycles $R_1$ and $R_2$, both sharing the same period of length $p(G)$ with $R(G)$.
Taking the connected unicyclic graphs $H_1$ and $H_2$ such that $R(H_1)\cong R_1$
and $R(H_2)\cong R_2$, we see that $G$ is CR-equivalent to the disjoint union of $H_1$ and $H_2$
and is, therefore, not CR-identifiable.
\end{proof}

\begin{lemma}\label{cor:UG=UH}
  Let $G$ and $H$ be connected unicyclic graphs with $c(H)\le c(G)$. Assume that both $G$ and $H$
  are CR-identifiable. Then $U^G\cong U^H$ if and only if
  $\Set{t(x)}{x\in \core(G)}=\Set{t(x)}{x\in \core(H)}$ and
  one of the following conditions is true:
  \begin{itemize}
  \item
    $G\cong H$,
  \item
    $p(G)=p(H)=1$ and $3\le c(H)<c(G)\le5$,
  \item
    $p(G)=p(H)=2$ and $c(H)=4$ while $c(G)=6$.    
  \end{itemize} 
\end{lemma}

\begin{proof}
( $\Leftarrow$ )
By Lemma \ref{cor:RGRH}.

\medskip

\noindent
( $\Rightarrow$ )
Let $G$ and $H$ be CR-identifiable connected unicyclic graphs with $c(H)\le c(G)$.
Assume that $U^G\cong U^H$. The equality $\Set{t(x)}{x\in \core(G)}=\Set{t(x)}{x\in \core(H)}$
immediately follows from Lemma \ref{cor:RGRH}. By the same corollary, $p(G)=p(H)=p$.
If $p\ge3$, then Lemma \ref{thm:unicyclic} yields the equality $c(G)=p(G)=p(H)=c(H)$,
which readily implies $G\cong H$ by using Lemma \ref{cor:RGRH} once again.
If $p\le2$, then either $c(G)=c(H)$ and $G\cong H$ or $c(H)<c(G)$ and then, by Lemma \ref{thm:unicyclic},
$c(H)$ and $c(G)$ are as claimed.
\end{proof}

\subsection{A general criterion of CR-identifiability}\label{ss:criterium}

Deciding whether a given graph is CR-identifiable is an efficiently solvable problem \cite{ArvindKRV17,KieferSS22}.
For our purposes, it is beneficial to have a more explicit description of CR-identifiable
graphs in terms of the complex and the simple part of a graph. We now derive such a description
from the facts obtained for unicyclic graphs in the preceding subsection.

\begin{theorem}\label{thm:amenability}\hfill
  \begin{enumerate}[\bf 1.]
  \item
    A graph $G$ is CR-identifiable if and only if both the complex and the simple parts of $G$
    are CR-identifiable.
  \item
    The simple part of $G$ is CR-identifiable if and only if both of the following two conditions are true:
  \begin{enumerate}[\bf (a)]
  \item every unicyclic component of $G$ is CR-identifiable, i.e., is as described in Lemma~\ref{thm:unicyclic};
  \item every two unicyclic components of $G$ have non-isomorphic universal covers, i.e.,
    there is no pair of connected components as described in Lemma~\ref{cor:UG=UH}.
  \end{enumerate}
  \end{enumerate}
\end{theorem}

\begin{proof}
\textit{1.}  
If $G$ is CR-identifiable, then its complex and simple parts are both CR-identifiable
as a consequence a more general fact: The vertex-disjoint union of any set of connected
components of $G$ is CR-identifiable. This fact is easy to see directly, and it also immediately
follows from Lemma~\ref{lem:union-amenable}.

In the other direction, assuming that the complex and the simple parts of $G$ are CR-identifiable,
we have to conclude that $G$ is CR-identifiable. Lemma \ref{lem:union-amenable} reduces our task
to verification that if $H$ is a complex connected component of $G$ and $S$ is a simple
connected component of $G$ (a tree or a unicyclic graph), then the universal covers of $H$ and $S$
are non-isomorphic. The last condition follows from the fact that the universal cover of a tree
is the tree itself and from Claim~\ref{lem:unique-path}.

\textit{2.}
The second part of the theorem follows immediately from Lemma \ref{lem:union-amenable} due to
the well-known fact \cite{Immerman1990} that every tree is CR-identifiable.
\end{proof}

\subsection{Coloring the cores of general graphs}\label{ss:CR-cores}

We here collect useful general facts about the CR-colors of vertices in the core of a graph.
Let $G$ be an arbitrary graph. If $x$ is a vertex in $\core(G)$, then in $G$ we have a tree growing from the root $x$
that shares with $\core(G)$ only the vertex $x$. We denote this rooted tree by~$T_x$.

\begin{claim}\label{lem:CGCH}
  Let $G$ and $H$ be graphs. Let $x$ be a vertex in $\core(G)$ and $y$ be a vertex in $\core(H)$.
  If $T_x\not\cong T_y$, then $C(x)\ne C(y)$.
\end{claim}

\begin{proof}
Clearly, it suffices to prove this for connected $G$ and $H$.
  The condition $T_x\not\cong T_y$ readily implies that $U^G_x\not\cong U^H_y$,
  and the claim follows from Lemma~\ref{lem:UvsC}.
\end{proof}

\begin{claim}
  Let $G$ and $H$ be two graphs (it is not excluded that $G=H$). For vertices $u\in V(G)$ and $v\in V(H)$
  assume that $C(u)=C(v)$. Then  $u\in\core(G)$ if and only if $v\in\core(H)$.
\label{cl:trees_cores_distinguish}
\end{claim}

\begin{proof}
Assume that $G$ and $H$ are connected (the general case will easily follow).
  Let $\alpha_G$ be a covering map from $U^G$ to $G$, and $\alpha_H$ be a covering map from $U^H$ to $H$.
  Consider $x\in V(U^G)$ and $y\in V(U^H)$ such that $\alpha_G(x)=u$ and $\alpha_H(y)=v$.
  Note that $u\in\core(G)$ if and only if there is a cycle in $U^G$ containing $x$,
  and the same is true about $v$ any $y$. This proves the claim because $U^G_x\cong U^H_y$ by Lemma~\ref{lem:UvsC}.
\end{proof}

In our proofs, we will deal with cores that locally have a tree structure, that is, the balls of sufficiently large radii
around most of its vertices induce trees. In this case, CR distinguishes vertices that have non-isomorphic neighborhoods.

\begin{claim}
  Let $B_r(v)$ denote the set of vertices at distance at most $r$ from a vertex $v$. 
  Let $v_1,v_2\in V(G)$. If, for some $r$, the $r$-neighborhoods $B_r(v_1)$ and $B_r(v_1)$ induce non-isomorphic trees
  rooted in $v_1$ and $v_2$ respectively, then $C(v_1)\neq C(v_2)$.
\label{cl:locally_trees_CR}
\end{claim}

\begin{proof}
  This is a direct consequence of Lemma~\ref{lem:UvsC}.
\end{proof}

Throughout the paper, we identify the vertex set of the kernel of $G$ with the set of vertices of $\core(G)$
having degrees at least 3 in the core. We now state another consequence of Lemma~\ref{lem:UvsC}.

\begin{claim}
Let $G$ be a graph with minimum degree at least 2 and let $\mathrm{K}$ be its kernel. Let $r$ be a positive integer. For $v\in V(\mathrm{K})$, let $\mathcal{B}^{\mathrm{K}}_r(v)$ be the subgraph of $\mathrm{K}$ induced by the set of vertices at distance at most $r$ from $v$ in $\mathrm{K}$. Let $\mathcal{B}_r(v)\subset G$ be the subdivided version of $\mathcal{B}^{\mathrm{K}}_r(v)$.
Let $v_1,v_2$ be vertices of $\mathrm{K}$ such that, for some $r$, graphs $\mathcal{B}_r(v_1),\mathcal{B}_r(v_2)\subset G$ are non-isomorphic trees rooted in $v_1,v_2$. Then $C^G(v_1)\neq C^G(v_2)$.
\label{cl:locally_trees_CR_2}
\end{claim}

Finally, we need the fact that the partition produced by CR on a graph refines the partition produced by CR on its core.

\begin{claim}
  Let $u$ and $v$ be vertices in $\core(G)$. Let $C$ and $C'$ be the colorings produced by CR run on $G$
  and $\core(G)$ respectively. If $C'(u)\neq C'(v)$, then also $C(u)\neq C(v)$.
\label{cl:CR_from_core_to_graph}
\end{claim}

\begin{proof}
Clearly, it is enough to prove the claim for a connected graph $G$.
Let us assume towards a contradiction that $C(u)=C(v)$. 
Let $\alpha$ be a covering map from $U^G$ to $G$. Let $x,y\in V(U^G)$ be such that $\alpha(x)=u$ and $\alpha(y)=v$.
Due to Lemma~\ref{lem:UvsC},  $U^G_x\cong U^G_y$. 
Therefore, $U^{\core(G)}_x=\core(U^G_x)\cong \core(U^H_y)=U^{\core(H)}_y$.
But then, again by Lemma~\ref{lem:UvsC}, $C'(u)=C'(v)$, a contradiction.
\end{proof}

\subsection{Derivation of Corollary \ref{cor:ident} from Theorem \ref{thm:ColRef}}\label{ss:proof-cor}

\textit{Part 1.}
Any isomorphism of graphs obviously respects their cores; cf.\ Claim~\ref{cl:trees_cores_distinguish}.
Claim \ref{lem:CGCH} shows that the CR-color of any vertex $x$ in the core $\coren$
contains a complete information about the isomorphism type of the rooted tree $T_x$
``growing'' from this vertex. This has the following consequence. Let $\corencol$
denote the colored version of $\coren$ where each vertex $x$ is colored by the isomorphism type of $T_x$.
Then $\compn$ is CR-identifiable if and only if $\corencol$ is CR-identifiable.
In order to show that $\corencol$ is CR-identifiable it suffices to show that $\corencol$ is
reconstructible up to isomorphism from the multiset of the vertex colors produced by CR
on input $\corencol$. Note that the CR-color partition of $\corencol$ is equal to the restriction
of the CR-color partition of $\compn$ to $\coren$ (recall Claim~\ref{cl:trees_cores_distinguish}).
Theorem \ref{thm:ColRef}, therefore, provides us with the following information (whp):
\begin{itemize}
\item
  every CR-color class of $\corencol$ has size either 1 or 2,
\item
  every two equally colored vertices have degree~2,
\item
  every two equally colored vertices form an orbit of $\aut(\corencol)$.
\end{itemize}
Moreover, our Main Lemma (Lemma~\ref{th:giant_main}) ensures that $\compn$ whp has no involutory automorphism of type $A_3$ described in Section \ref{sc:intro}. Along with this fact, the
above conditions readily imply that the color classes of size 2 occur either
``along'' a pair of transposable pendant paths between two vertices of degree at least 3
or correspond to the reflectional symmetry of a pendant cycle.
Here we use the notions introduced in Section \ref{sc:intro} in the context of $\aut(\coren)$,
which should now be refined by taking into account the coloring of~$\corencol$.

If $\{u\}$ and $\{v\}$ are two color classes of size 1, then the colors $C(u)$ and $C(v)$
yield the information on whether the vertices $u$ and $v$ are adjacent or not.
For color classes $\{u\}$ and $\{v,v'\}$, note that $u$ and $v$ are adjacent
if and only if $u$ and $v'$ are adjacent. This adjacency pattern is as well
reconstructible from the colors $C(u)$ and $C(v)=C(v')$. If $\{u,u'\}$ and $\{v,v'\}$ are
two color classes of size 2, then they span either a complete or empty bipartite graph
or a matching (for example, $u$ is adjacent to $v$, $u'$ is adjacent to $v'$,
and there is no other edges between these color classes).
Each of these three possible adjacency patterns is reconstructible from the
colors $C(u)=C(u')$ and $C(v)=C(v')$. A crucial observation, completing the proof,
is that all ways to put a matching between $\{u,u'\}$ and $\{v,v'\}$ lead to
isomorphic graphs.

\textit{Part 2.} This follows from part 1 by part 1 of Theorem~\ref{thm:amenability}.

\subsection{Derivation of Theorem \ref{thm:CanLab} from Theorem \ref{thm:ColRef}}
\label{sc:from_CR_to_CL}

Before proceeding to the proof, we remark that when we say that a canonical labeling algorithm
succeeds on a random graph $G_n$, we mean that the algorithm works correctly on a certain
efficiently recognizable (closed under isomorphisms) class of graphs $\mathcal{C}$ such that
$G_n$ belongs to $\mathcal{C}$ whp. Though not explicitly stated in the argument below,
it will be clear that, in our case, $\mathcal{C}$ is the class of all graphs satisfying the
conditions of Theorem \ref{thm:ColRef}. Note that these conditions are easy to check after
running CR on a graph.

First of all, we distinguish the complex and the simple parts of $G_n$
and compute a canonical labeling of the simple part separately.
This is doable in linear time. It remains to handle the complex part~$\compn$.

It is enough to compute a suitable injective
coloring of $\compn$ and subsequently to rename the colors in their lexicographic order
by using the labels that were not used for the simple part.
To this end, we run CR on $\compn$. This takes time $O(n\log n)$
as CR can be implemented \cite{BBG} in time $O((n+m)\log n)$, where $m$ denotes the number of edges
(which is linear for the sparse random graph under consideration).
Then we begin with coloring the vertices of the core $\coren$.
Theorem \ref{thm:ColRef} along with Claim~\ref{cl:trees_cores_distinguish} ensures that
the vertices of degree at least 3 already received individual colors.
The duplex colors occur along transposable pendant paths and pendant cycle
(like in Section \ref{ss:proof-cor}, these notions are understood with respect to $\compn$
rather than to $\coren$ alone). To make such vertex colors unique, we keep the original colors
along one of two transposable paths and concatenate their counterparts in the other path
with a special symbol. We proceed similarly with symmetric pendant cycles. In this way,
every vertex $x$ in the core $\coren$ receives an individual color $\ell(x)$.
In the last phase, we compute a canonical labeling for each tree part $T_x$ of $\compn$,
regarding $T_x$ as a tree rooted at $x$. This coloring is not injective yet because
some $T_x$ and $T_y$ can be isomorphic. This is rectified by concatenating all vertex
colors in $T_x$ with~$\ell(x)$.

\section{CR-coloring of the random graph}
\label{cs:rg_amenability}

In this section, we state and prove our Main Lemma that completely describes the output of CR on the random graph. 
We need the following definition. Given a graph $G$, we call vertices $u$ and $v$ in $\mathrm{core}(G)$ {\it interchangeable}, if
\begin{itemize}
\item they both have degree 2 in $\mathrm{core}(G)$,
\item $u$ and $v$ belong to a cycle $F\subset\mathrm{core}(G)$ with the following property: there exists a vertex $w$ on the cycle such that $w$ has degree at least 3 in $\mathrm{core}(G)$, $d_F(u,w)=d_F(v,w)$, and all the other vertices on the cycle, but the vertex opposite to $w$ when $|V(F)|$ is even, have degree 2 in $\mathrm{core}(G)$. In other words, $u$ and $v$ either belong to a pendant cycle or to two transposable pendant paths, and the respective transposition replaces $u$ and $v$.
\end{itemize}

\begin{lemma}[Main Lemma]\label{th:giant_main}
 Let $\gamma>1$ be a constant, $pn\leq \gamma$, and $G_n=G(n,p)$. Let $\mathrm{H}_n$ be the union of complex components in $G_n$, and $\mathrm{C}_n$ be its core.
If CR is run on $\mathrm{H}_n$, then whp any pair of vertices in $\mathrm{C}_n$ receiving the same color is interchangeable. Under the condition $pn=1+\omega(n^{-1/3})$, this is true also if CR is run on $\mathrm{C}_n$.
\end{lemma}

In the next section, we recall and prove several basic facts about the distribution of connected components in $G(n,p)$ that we use further in Section~\ref{sc:proof_main} to prove Main Lemma. In Section~\ref{sc:main_critical}, we state and prove Lemma~\ref{lm:critical_CR_core} that complements Main Lemma in the case $pn=1+O(n^{-1/3})$. We will use both Main Lemma and Lemma~\ref{lm:critical_CR_core} to derive Theorems~\ref{thm:ColRef}~and~\ref{th:symmetries} in Section~\ref{sc:proofs_th_Intro}.

In this section, we use the following notations. For a graph $G$, $d_G(u,v)$ is the shortest-path distance between $u$ and $v$ in $G$. Sometimes, when the graph is clear from the context, we omit the subscript $G$. For a vertex $v$ and a real number $r$, we denote by $\mathcal{B}^G_{r}(v)$ the ball of radius $r$ around $v$ in $G$ --- i.e. the graph induced on the set of all vertices at distance at most $r$ from $v$ in $G$. For a non-negative integer $r$, we denote by $\mathcal{S}^G_{r}(v)\subset\mathcal{B}^G_{r}(v)$ the sphere of radius $r$ around $v$ in $G$ --- i.e. the graph induced on the set of all vertices at distance exactly $r$ from $v$ in $G$. 

For a connected graph $G$, its {\it excess} is the difference between the number of edges and the number of vertices. In particular, a tree has excess $-1$. We call {\it $\ell$-complex} a connected graph with excess $\ell$. The {\it total excess} of a graph without unicyclic components is the sum of excesses of all its components.

\subsection{Structure of complex components in the random graph}
\label{sc:complex_structure}

Let us recall that, for any constant $\varepsilon>0$ and any $p>\frac{1+\varepsilon}{n}$, whp the random graph $G_n\sim G(n,p)$ has a unique {\it giant} component of linear in $n$ size and all the other components are unicyclic and have sizes $O(\log n)$. When $1+\omega(n^{-1/3})=np\leq 1+\varepsilon$, still whp there is a unique complex component of size $\Theta((np-1)n)$.
On the other hand, if $np=1-\omega(n^{-1/3})$, then whp all components of $G_n\sim G(n,p)$ are unicyclic (for more details, see, e.g.,~\cite[Chapter 5]{Janson_book}). In particular, in the latter case whp there are no complex components and there is nothing to prove. So, in what follows, we assume that $np\geq 1-O(n^{-1/3})$. We will also need the following simple claim.

\begin{claim}
Let $\gamma>1$, $np\leq\gamma$, and $G_n\sim G(n,p)$. 
  There exists $\varepsilon=\varepsilon(\gamma)$ such that $\varepsilon\to\infty$ as $\gamma\to 1$ and whp any connected subgraph of $G_n$ of size at most $\varepsilon\ln n$ is not complex.

\label{cl:complex_components_sizes}
\end{claim}

\begin{proof}
Let $k\in\mathbb{Z}_{>0}$. Let $\mathcal{C}_k$ be the family of all connected 1-complex graphs on $[k]$ with minimum degree 2. Clearly, such graphs can only be constructed as follows: take two cycles on two subsets $V_1,V_2\cup[k]$ such that $V_1\cup V_2=[k]$ and either $V_1\cap V_2=\varnothing$, or the two cycles share a single path. If the two cycles are edge-disjoint, then add a path joining them. It is easy to see that 1) if a connected graph is complex, then it has a subgraph from $\cup_k\mathcal{C}_k$; 2) $|\mathcal{C}_k|\leq k!\cdot k^2$. The expected number of subgraphs in $G(n,p)$ of size at most $N$ that have an isomorphic copy in $\cup_k\mathcal{C}_k$ is at most
$$
\sum_{k\leq N}{n\choose k}k! k^2 p^{k+1}\leq \sum_{k\leq N}k^2 n^k p^{k+1}.
$$
In particular, if $p\leq \gamma/n$, then the latter expression is $O(N^2 \gamma^N/n)=o(1)$ for an appropriate choice of $\varepsilon=\varepsilon(\lambda)$, implying the desired assertion. 
\end{proof}

Also, we will use the following contiguous models for the giant components in supercritical random graphs.

\begin{theorem}[J. Ding, E. Lubetzky, Y. Peres~\cite{DLP_anatomy}]
Assume that $pn=\lambda$ for some constant $\lambda>1$. Let $\mathrm{H}_n$ be the largest component of $G_n\sim G(n,p)$. Let $\mu$ be the unique number in $(0,1)$ such that $\mu e^{-\mu}=\lambda e^{-\lambda}$.  Let $\xi\sim\mathcal{N}\left(\lambda-\mu,1/n\right)$. Consider a random vector $(\eta_1,\ldots,\eta_n)$ consisting of independent $\mathrm{Pois}(\xi)$ random variables. Denote by $\mathcal{A}$ the event that $\sum\eta_i\1_{\eta_i\geq 3}$ is even. For $i\in[n]$, let $D_i=\frac{1}{\mathbb{P}(\mathcal{A})}\eta_i \1_{\mathcal{A}}$. For every $k\geq 3$, set $n_k=\sum_i \1_{D_i=k}$, $N=\sum_{k\geq 3}n_k$.

Let $\mathrm{\tilde H}_n$ be a random graph constructed in the following way:
\begin{itemize}
\item Define a uniformly random multigraph $\mathrm{\tilde K}_n$ on $[N]$ such that, for every $k\geq 3$, it has exactly $n_k$ vertices of degree $k$.
\item For every edge of $\mathrm{\tilde K}_n$, generate independently a $\mathrm{Geom}(1-\mu)$ random variable $\zeta$, and replace this edge with a path of length $\zeta$.
\item Attach an independent $\mathrm{Pois}(\mu)$--Galton--Watson tree to each vertex of the obtained at the previous step random graph.
\end{itemize}
Then, for every set of graphs $\mathcal{F}$, if $\mathbb{P}(\mathrm{\tilde H}_n\in\mathcal{F})=o(1)$, then $\mathbb{P}(\mathrm{H}_n\in\mathcal{F})=o(1)$ as well.
\label{th:contiguous_super_critical}
\end{theorem}

\begin{theorem}[J. Ding, J. H. Kim, E. Lubetzky, Y. Peres~\cite{DKLP:anatomy}]
Assume that $pn=1+\delta_n$, where $n^{-1/3}\ll\delta_n=o(1)$, $G_n\sim G(n,p)$. Let $\mathrm{H}_n$ be the largest component of $G_n$. Let $\mu$ be the unique number in $(0,1)$ such that $\mu e^{-\mu}=pn e^{-pn}$.  Let $\xi\sim\mathcal{N}\left(pn-\mu,\frac{1}{n\delta_n}\right)$. Consider a random vector $(\eta_1,\ldots,\eta_n)$ consisting of independent $\mathrm{Pois}(\xi)$ random variables. Denote by $\mathcal{A}$ the event that $\sum\eta_i\1_{\eta_i\geq 3}$ is even. For $i\in[n]$, let $D_i=\frac{1}{\mathbb{P}(\mathcal{A})}\eta_i \1_{\mathcal{A}}$. For every $k\geq 3$, set $n_k=\sum_i \1_{D_i=k}$, $N=\sum_{k\geq 3}n_k$.
Let $\mathrm{\tilde H}_n$ be a random graph constructed in the following way:
\begin{itemize}
\item Define a uniformly random multigraph $\mathrm{\tilde K}_n$ on $[N]$ such that, for every $k\geq 3$, it has exactly $n_k$ vertices of degree $k$.
\item For every edge of $\mathrm{\tilde K}_n$, generate independently a $\mathrm{Geom}(1-\mu)$ random variable $\zeta$, and replace this edge with a path of length $\zeta$.
\item Attach an independent $\mathrm{Pois}(\mu)$--Galton--Watson tree to each vertex.
\end{itemize}
Then, for every set of graphs $\mathcal{F}$, if $\mathbb{P}(\mathrm{\tilde H}_n\in\mathcal{F})=o(1)$, then $\mathbb{P}(\mathrm{H}_n\in\mathcal{F})=o(1)$ as well.
\label{th:contiguous_critical}
\end{theorem}

From these two theorems, we derive the following.

\begin{claim}
Let $\delta>0$ be a constant, $n^{-1/3}\ll\delta_n:=pn-1\leq\delta$, and $G_n\sim G(n,p)$. Then whp in $K(G_n)$ there are no complex subgraphs of size at most $(\ln(n\delta^3_n))^{3/4}$.
\label{cl:kernel_complex}
\end{claim}

\begin{proof}
Here, we consider the contiguous models from Theorems~\ref{th:contiguous_super_critical}~and~\ref{th:contiguous_critical}. By the union bound, the maximum degree in $\mathrm{\tilde K}_n$ is at most 
$$
\Delta:=\frac{2\ln n}{\ln(1+1/\delta_n)}.
$$
Also, whp it has $N=\Theta(n\delta_n^3)$ vertices.
Expose $\eta_1,\ldots,\eta_n$ and define the degree sequence of the kernel. Assume that the maximum degree is at most $\Delta$ and the number of vertices is $N=\Theta(n\delta_n^3)$. Let $\mathcal{G}$ be the set of multigraphs on $[N]$ with the specified degree sequence. Let 
$$
d:=\left\lceil(\ln(n\delta^3_n))^{3/4}\right\rceil.
$$
Due to Theorems~\ref{th:contiguous_super_critical}~and~\ref{th:contiguous_critical}, it is enough to show that whp there are no complex subgraphs of size at most $d$ in $\mathrm{\tilde K}_n$.

Fix a vertex $v$ in the kernel.
For any tree $T$ on a subset of $[N]$ of depth $d$ rooted in $v$ and maximum degree at most $\Delta$, let us consider the family $\mathcal{G}_T$ of all multigraphs in $\mathcal{G}$ such that $T$ is the full BFS tree of depth $d$ emanating from $v$. Fix two pairs of vertices $\{x_1,y_1\},\{x_2,y_2\}\in (V(T))^2$. Let us consider a partition $\mathcal{G}_T^0\sqcup\mathcal{G}_T^1=\mathcal{G}_T$ such that all multigraphs in $\mathcal{G}_T^1$ have both edges $\{x_1,y_1\}$ and $\{x_2,y_2\}$, while all multigraphs in $\mathcal{G}_T^0$ do not have neither of these edges.

Take any $G\in \mathcal{G}_T^1$. Note that it has at least $3N/2-\Delta^{d+2}>N$ edges that are entirely outside $V(T)$. Let $\{x'_1,y'_1\}$ and $\{x'_2,y'_2\}$ be two such edges. Clearly, we may switch the edges
$$
  \{x_1,y_1\},\{x_2,y_2\},\{x'_1,y'_1\},\{x'_2,y'_2\}\,\rightarrow\,\{x_1,x'_1\},\{x_2,x'_2\},\{y_1,y'_1\},\{y_2,y'_2\}
$$
and get the multigraph with edges $\{x_1,x'_1\},\{x_2,x'_2\},\{y_1,y'_1\},\{y_2,y'_2\}$ from $\mathcal{G}_T^0$. Thus, taking into account repetitions, $G$ corresponds to at least $\frac{N(N-1)}{2\Delta^2}$ multigraphs from $\mathcal{G}_T^0$. On the other hand, each multigraph from $\mathcal{G}_T^0$ corresponds to at most $\Delta^4$ multigraphs from $\mathcal{G}_T^1$ with edges $\{x_1,y_1\}$ and $\{x_2,y_2\}$. Indeed, for every $G\in\mathcal{G}_T^0$ we can switch at most $\Delta^2$ pairs of edges going from $x_1$ and $y_1$ to $[N]\setminus V(T)$ and at most $\Delta^2$ pairs of edges going from $x_2$ and $y_2$ to $[N]\setminus V(T)$.

We conclude that 
$$
|\mathcal{G}_T^1|\leq\frac{2\Delta^6}{n'(n'-1)}|\mathcal{G}_T^0|.
$$
Note that the number of choices of the pair $\{x_1,x_2\},\{y_1,y_2\}$ is at most $|V(T)|^4\leq\Delta^{4d+4}$.
Denoting by $\mathbf{T}_v$ the full BFS subtree of depth $d$ in $\mathrm{\tilde K}_n$ rooted in $v$, we conclude that the probability that in $\mathrm{\tilde K}_n$ there is a complex subgraph of size at most $d$ is at most
$$
 \sum_v\sum_T\mathbb{P}(\mathrm{\tilde K}_n\in\mathcal{G}_T^1\mid \mathbf{T}_v=T)\mathbb{P}(\mathbf{T}_v=T)\leq N\Delta^{4d+4}\frac{2\Delta^6}{N(N-1)}=o(1),
$$
completing the proof.
\end{proof}

\subsection{Proof of Main Lemma (Lemma~\ref{th:giant_main})}
\label{sc:proof_main}

 We consider separately large $p$ (supercritical phase) and small $p$ (critical phase). In both cases, we specify good sets of vertices and show that all vertices from good sets are distinguished by CR. This is done in Section~\ref{sc:proof2_large} for large $p$ and Section~\ref{sc:proof2_small} for small $p$. Finally, in Section~\ref{sc:proof3} we complete the proof:
  we show that distinguishing between good vertices in the core is sufficient to distinguish between all pairs in the core that are not interchangeable.

\subsubsection{Distinguishing good vertices in the core in the supercritical and strictly supercritical phases}
\label{sc:proof2_large}

In this section, we let $p=p(n)$ be such that $\gamma\geq np=1+\omega(n^{-1/3})$ for some constant $\gamma>1$. Denote $\delta_n:=np-1$. We also let $\mathrm{K}_n$ and $\mathrm{C}_n$ be the kernel and the core of the giant component of $G_n\sim G(n,p)$. Let $C^{\mathrm{core}}$ be the coloring produced by CR on $\mathrm{C}_n$.

We assign to every edge $e$ of $\mathrm{C}_n$ the weight $1/\ell$, where $\ell-1$ is the number of vertices that subdivide the edge of the kernel $e$ belongs to. The weight of a path is the sum of weights of its edges. For $u,v\in V(\mathrm{C}_n)$, let $d^f(u,v)$ be the {\it fractional distance} between $u$ and $v$, i.e. the minimum weight of a path between $u$ and $v$. We denote the respective metric space by $\mathcal{M}_n$.

Fix a positive real $s$. Let $D_s$ be the set of all $v\in V(\mathrm{C}_n)$ such that the ball around $v$ in $\mathcal{M}_n$ of radius $s$ induces an acyclic graph.
 For every vertex $v\in D_s$ and  integer $r<s$, let $\mathcal{P}_{r}(v)$ be the multiset of lengths of edge-disjoint paths from $\mathrm{C}_n$ that are produced by subdividing edges $\{x,y\}\in E(\mathrm{K}_n)$, where $d^f(x,v)\leq r$ while $d^f(y,v)>r$.
 We will require the following assertions.

\begin{claim}
Let $s\geq 0.6(\ln(\delta^3_n n))^{2/3}$. Whp for any two different $u,v\in D_s\cap V(\mathrm{K}_n)$, there exists an integer $r\leq 0.5(\ln(\delta^3_n n))^{2/3}$ such that the multisets $\mathcal{P}_{r}(u)$ and $\mathcal{P}_{r}(v)$ are different.
\label{cl:supercritical_main}
\end{claim}

\begin{proof}
We fix $s\geq 0.6(\ln(\delta^3_n n))^{2/3}$ and let $D:=D_s$. We prove this claim in the contiguous models $\tilde G_n$, defined in Theorems~\ref{th:contiguous_super_critical},~\ref{th:contiguous_critical},  and then use these theorems to conclude that it also holds in $G_n$. So, in what follows, $\mathrm{\tilde K}_n=\mathrm{K}(\tilde G_n)$, $\mathrm{\tilde C}_n=\mathrm{C}(\tilde G_n)$,  and $\tilde D=D(\tilde K_n)$.

Let us expose $\mathrm{\tilde K}_n$ and let $u,v\in \tilde D\cap V(\mathrm{\tilde K}_n)$. Recall that whp $N=|V(\tilde K_n)|=\Theta(\delta^3_n n)$ due to Theorems~\ref{th:contiguous_super_critical},~\ref{th:contiguous_critical}. Assume first that the distance between $u$ and $v$ is at most $0.4(\ln (\delta^3_n n))^{2/3}$ in $\mathrm{\tilde K}_n$. Let $P$ be the shortest path between $u$ and $v$ --- it is unique due to the definition of $\tilde D$. Let $v'$ be a neighbor of $v$ in $\mathrm{\tilde K}_n$ that does not belong to $P$. Then, by the definition of $\tilde D$, 
$$
\left|\mathcal{S}^{\mathrm{\tilde K}_n}_r(v')\setminus\mathcal{B}^{\mathrm{\tilde K}_n}_r(v)\right|\geq 2^{r}
$$ 
for all $r\in \left[0.4(\ln (\delta^3_n n))^{2/3},0.5(\ln (\delta^3_n n))^{2/3}\right]$. Since $u,v\in\tilde D$, we have that $\mathcal{B}^{\mathrm{\tilde K}_n}_{s}(u)$ and $\mathcal{B}^{\mathrm{\tilde K}_n}_{s}(v)$ are trees. It immediately implies, that for every such~$r$, 
$$
\left|\mathcal{S}^{\mathrm{\tilde K}_n}_{r+1}(v)\setminus\mathcal{B}^{\mathrm{\tilde K}_n}_{r+1}(u)\right|\geq 2^{r}.
$$

We then generate subdivisions of the edges of the kernel from the definition of $\tilde G_n$ in the following order: for every $r=\left\lceil 0.4(\ln (\delta^3_n n))^{2/3}\right\rceil,\ldots,\left\lfloor 0.5(\ln (\delta^3_n n))^{2/3}\right\rfloor$, we, first, subdivide all edges growing from $\mathcal{B}^{\mathrm{\tilde K}_n}_{r+1}(u)$ outside of the ball, and then all edges growing from $\mathcal{S}^{\mathrm{\tilde K}_n}_{r+1}(v)$ outside of $\mathcal{B}^{\mathrm{\tilde K}_n}_{r+1}(v)$. Notice that all sets $\mathcal{S}^{\mathrm{\tilde K}_n}_{r+1}(v)$ are disjoint for different $r$. For every $r$, as soon as the edges that correspond to the vertex $u$ are subdivided, the event that $\mathcal{P}_{r+1}(u)=\mathcal{P}_{r+1}(v)$ immediately implies that the {\it random} multiset of lengths of paths from $\mathrm{\tilde C}_n$, that are produced by subdividing edges from $\mathrm{\tilde K}_n$ that grow from $\mathcal{B}_{r+1}^{\mathrm{\tilde K}_n}(v)$ outside, should be equal to a predefined value. This multiset has size at least $2^r$. Since the geometric random variables considered in Theorems~\ref{th:contiguous_super_critical},~\ref{th:contiguous_critical} do not have atoms with probability measure $1-o(1)$, the latter event has probability at most $2^{-\Theta(r)}$ due to the de Moivre--Laplace local limit theorem. Eventually, 
\begin{multline*}
\mathbb{P}\left(\mathcal{P}_{r+1}(u)=\mathcal{P}_{r+1}(v)\text{ for all 
}r\in\left[0.4(\ln (\delta^3_n n))^{2/3},0.5(\ln (\delta^3_n n))^{2/3}\right]\right)\leq\\
\leq\exp\left(-\Theta((\log (\delta^3_n n))^{4/3})\right).
\end{multline*}

Assume now that the distance between $u$ and $v$ is bigger than $0.4(\ln (\delta^3_n n))^{2/3}$ in $\mathrm{\tilde K}_n$. Then, by the definition of $\tilde D$, sets $\mathcal{B}^{\mathrm{\tilde K}_n}_{0.2(\ln (\delta^3_n n))^{2/3}}(v)$ and $\mathcal{B}^{\mathrm{\tilde K}_n}_{0.2(\ln (\delta^3_n n))^{2/3}}$ are disjoint and sets $\mathcal{S}^{\mathrm{\tilde K}_n}_{r}(v)$ have size at least $2^{r}$ for all $r\in[0.15(\ln (\delta^3_n n))^{2/3},0.2(\ln (\delta^3_n n))^{2/3}-1]$. As above, we get that $\mathcal{P}_{r}(u)=\mathcal{P}_{r}(v)$ for all $r\in\left[0.15(\ln (\delta^3_n n))^{2/3},0.2(\ln (\delta^3_n n))^{2/3}-1\right]$ with probability at most $\exp\left(-\Theta((\log (\delta^3_n n))^{4/3})\right)$.

The union bound over all pairs $u,v\in \tilde D$ and Theorems~\ref{th:contiguous_super_critical},~\ref{th:contiguous_critical} complete the proof.
\end{proof}

We now let 
$$
s^*:=\lfloor(\ln(\delta^3_n n))^{2/3}\rfloor, \,\,
D:=D_{s^*}.
$$

\begin{claim}
Whp for any distinct $u,v\in D$, $C^{\mathrm{core}}(u)\neq C^{\mathrm{core}}(v)$.
\label{cl:good_distinguish_large_p}
\end{claim}

\begin{proof}
Assume that the assertion of Claim~\ref{cl:supercritical_main} holds for $s=s^*$ and $s=s^*-1$ deterministically. Let $u,v\in D\cap V(\mathrm{K}_n)$.  Let $B_u$ and $B_v$ be the subdivided versions of $\mathcal{B}_{s^*}^{\mathrm{K}_n}(u)$ and $\mathcal{B}_{s^*}^{\mathrm{K}_n}(v)$ in $\mathrm{C}_n$. Since $B_u\ncong B_v$ due to the conclusion of Claim~\ref{cl:supercritical_main}, we get that $C^{\mathrm{core}}(u)\neq C^{\mathrm{core}}(v)$ due to Claim~\ref{cl:locally_trees_CR_2}.

It remains to consider the case $v\in D\setminus V(\mathrm{K}_n)$ and $v\neq u\in D$. Assume towards contradiction that $C^{\mathrm{core}}(u)= C^{\mathrm{core}}(v)$. Then, both $u$ and $v$ have degree 2 in $\mathrm{C}_n$. In particular, $u\notin V(\mathrm{K}_n)$. Consider the edges $e_u,e_v$ of $\mathrm{K}_n$ that $u$ and $v$ subdivide. Let $P_u,P_v$ be the subdivided versions of $e_u,e_v$. Due to the assertion of Claim~\ref{cl:supercritical_main} applied to $s=s^*-1$, we get that all vertices of $\mathrm{K}_n$ from $e_u\cup e_v$ have different colors. On the other hand, by the definition of CR, the neighbors of $u$ should have exactly the same color as the neighbors of $v$. Thus, by induction, we get that the entire paths $P_u,P_v$ are colored identically. It may only happen if the endpoints of $P_u$ coincide with the endpoints of $P_v$. By the definition of $D$, it means that $P_u=P_v=:P$. Since the endpoints of $P$ are colored in different colors, it can be easily shown by induction that all vertices in $P$ are also colored in different colors. Thus, $u=v$ --- a contradiction.
\end{proof}

\subsubsection{Distinguishing good vertices in the core in the critical regime}
\label{sc:proof2_small}

Let $A$ be a large positive number. Let 
$$
1-n^{-1/3}\ln n\leq pn=1+o(1),\quad G_n\sim G(n,p).
$$
Whp any complex component in $G_n$ has size at least $100A\ln n$ due to Claim~\ref{cl:complex_components_sizes}.

Let us say that a path $u_1\ldots u_k$ {\it extends} the path $v_1\ldots v_k$ if, for some $i\in\{2,\ldots,k\}$ the sets $\{v_1,\ldots,v_{i-1}\}$ and $\{u_{k-i+2},\ldots,u_k\}$ are disjoint and $u_1=v_i,\ldots,u_{k-i+1}=v_k$. For convenience, we assume that this notion is closed under rotations of paths, i.e. if $u_1\ldots u_k$ extends $v_1\ldots v_k$, then it also extends $v_k\ldots v_1$ and we also say that $u_k\ldots u_1$ extends both $v_1\ldots v_k$ and $v_k\ldots v_1$ in this case. We call two paths $v_1\ldots v_k$ and $u_1\ldots u_k$ {\it weakly disjoint}, if they are either vertex-disjoint or one paths extends the other one.

\begin{claim}
Whp in $G_n$ there are no two weakly disjoint paths $v_1\ldots v_k$ and $u_1\ldots u_k$ of length $k=\lfloor A\ln n\rfloor$ such that, for every $i\in\{2,\ldots,k-1\}$, $v_i$ has degree 2 if and only if $u_i$ has degree 2.
\label{cl:weakly_disjoint_paths}
\end{claim}

\begin{proof}
Due to Claim~\ref{cl:complex_components_sizes}, whp in $G_n$ there are no complex subgraphs with at most $2k$ vertices.

For a path $P=v_1\ldots v_k$ in $G_n$, let us consider a binary word $w(P)=(w_2,\ldots,w_{k-1})$ defined as follows: $w_i=1$ if and only if $v_i$ has degree $2$ in $G_n$. Notice that, if a path $u_1\ldots u_k$ extends the path $v_1\ldots v_k$ so that $u_1=v_i,\ldots,u_{k-i+1}=v_k$ and $w(v_1\ldots v_k)=w(u_1\ldots u_k)$, then $w(u_1\ldots u_k)$ is periodic and defined by $w(v_1\ldots v_{i+1})=(w_2,\ldots,w_i)$. 

Let $X$ be the number of pairs of paths as in the statement of the claim and such that there are at most 2 edges between the paths (we are allowed to assume this since there are no complex subgraphs of size at most $2k$). Fix two weakly disjoint path $\mathbf{v}=v_1\ldots v_k$ and $\mathbf{u}=u_1\ldots u_k$ and assume without loss of generality that either $\mathbf{u}$ extends $\mathbf{v}$, or they are disjoint. Let $i$ be such that $u_1=v_i$. If there is no such $i$, i.e. the paths are disjoint, set $i=k+1$. Then, expose edges from all $v_j$, $j\leq i$, and assume that they send at most 2 edges to $u_2,\ldots,u_{k-1}$, other than the edge $\{u_1,u_2\}$. Then, probability that for every $j\in\{2,\ldots,k-1\}$, $v_j$ has degree 2 if and only if $u_j$ has degree 2, is at most 
$$
\max\left\{(1-p)^{n-2k},(1-(1-p)^n)\right\}^{k-4}\leq \left(1-e^{-(1+o(1))}\right)^{k-4}=o\left(\left(\frac{2}{3}\right)^k\right).
$$
We then get
$$
 \mathbb{E}X\leq\sum_{i=2}^{k+1}n^{k+i-1}p^{k+i-3}\left(\frac{2}{3}\right)^k\leq k n^2 (1+o(1))^{2k}\left(\frac{2}{3}\right)^k=o(1),
$$
for an appropriate choice of $A$. Due to Markov's inequality, $\mathbb{P}(X\geq 1)\leq\mathbb{E}X=o(1)$, completing the proof.
\end{proof}

Let $D$ be the set of all vertices $v\in \mathrm{C}_n=\mathrm{core}(G_n)$ that belong to a complex component of $G_n$ and such that $B_v:=\mathcal{B}^{G_n}_{3A\ln n}(v)$ is a tree. Let $C$ be the coloring produced by CR on $G_n$.

\begin{claim}
Whp for any $u,v\in D$, $C(u)\neq C(v)$.
\label{cl:good_distinguish_small_p}
\end{claim}

\begin{proof}
Assume that the statement of Claim~\ref{cl:weakly_disjoint_paths} holds deterministically in $G_n$. Fix two different $u,v\in D$. Let us show towards contradiction that trees $B_v$ and $B_u$ are not isomorphic. 

Take an arbitrary path $v_1\ldots v_k$ of length $k:= \lfloor 1.9 A\ln n\rfloor$, where $v_1=v$. Since, by assumption, $B_v\cong B_u$, there exists a path $\mathbf{u}=u_1\ldots u_k$ such that $u_1=u$ and, for every $i\in\{2,\ldots,k-1\}$, $v_i$ has degree 2 if and only if $u_i$ has degree 2. In the same way, since $v\in\mathrm{core}(G_n)$ and $B_v$ is a tree, we may consider a path $v'_1\ldots v'_k$ that shares only the vertex $v'_1=v=v_1$ with $v_1\ldots v_k$. Since $B_u\cong B_v$, there should be a path $\mathbf{u}'=u'_1\ldots u'_k$ that shares with $u_1\ldots u_k$ the only vertex $u'_1=u=u_1$ and such that, for every $i\in\{2,\ldots,k-1\}$, $v'_i$ has degree 2 if and only if $u'_i$ has degree 2. Since $B_v$ is acyclic and since pairs of paths $v_1\ldots v_k,\, u_1\ldots u_k$ and $v'_1\ldots v'_k,\, u'_1\ldots u'_k$ cannot be disjoint due to Claim~\ref{cl:weakly_disjoint_paths},  $u$ must lie on the path $P:=v_k\ldots v_1\ldots v_k'$. Moreover, since $B_u$ is acyclic, once $\mathbf{u}$ or $\mathbf{u}'$ leave $P$, they never meet with $P$ again. Thus, the path $u_k\ldots u_1\ldots u_k'$ is divided by $P$ in at most 3 parts: the first part does not have common vertices with $P$, the second part is a subpath of $P$, and the third part does not have common vertices with $P$ again. Let $Q$ be the longest part of the three. Then $Q$ has length $\ell\geq\frac{1}{3}(2k-1)>A\ln n$. Moreover, since $\mathrm{deg}_{G_n}u=\mathrm{deg}_{G_n}v$ by assumption and $u\neq v$, there should be a subpath $P'\subset P$ such that $P'$ and $Q$ are weakly disjoint, and the degrees of internal vertices in $P'$ and $Q$ are aligned in the sense that the $i$-th inner vertex of $P'$ have degree 2 if and only if the $i$-th vertex of $Q$ has degree 2. This is impossible due to Claim~\ref{cl:weakly_disjoint_paths}.

Thus, $B_v\ncong B_u$ implying $C(u)\neq C(v)$ due to Claim~\ref{cl:locally_trees_CR}.
\end{proof}

\subsubsection{Completing the proof of Main Lemma (Lemma~\ref{th:giant_main})}
\label{sc:proof3}

Due to Claim~\ref{cl:CR_from_core_to_graph}, Claim~\ref{cl:good_distinguish_large_p}, and Claim~\ref{cl:good_distinguish_small_p} in order to prove Main Lemma, it remains to prove the following:
\begin{enumerate}
\item for $1+\omega(n^{-1/3})= pn\leq \gamma$,
\begin{itemize} 
\item for every  $v\in V(\mathrm{C}_n)\setminus D$ and $u\in D$, $C^{\mathrm{core}}(u)\neq C^{\mathrm{core}}(v)$; 
\item for any not interchangeable pair of different vertices $u,v\in V(\mathrm{C}_n)\setminus D$, $C^{\mathrm{core}}(u)\neq C^{\mathrm{core}}(v)$;
\end{itemize}
\item for $1-n^{-1/3}\ln n\leq pn=1+o(1)$,
\begin{itemize}
\item for every  $v\in V(\mathrm{C}_n)\setminus D$ and $u\in D$, $C(u)\neq C(v)$; 
\item for any not interchangeable pair of different vertices $u,v\in V(\mathrm{C}_n)\setminus D$, $C(u)\neq C(v)$.
\end{itemize}
\end{enumerate}

For the sake of brevity, below we prove both statements in two different regimes simultaneously. Thus, with some abuse of notation, in the supercritical phase (i.e., $1+\omega(n^{-1/3})= pn\leq \gamma$), since we only consider CR on $\mathrm{C}_n$, we let $G_n:=\mathrm{C}_n$ and $C:=C^{\mathrm{core}}$. We also assume that, when $1+\omega(n^{-1/3})= pn\leq \gamma$, the core is equipped with the fractional distance $d^f$, constituting the metric space $\mathcal{M}_n$. If $1-n^{-1/3}\ln n\leq pn\leq 1+o(1)$, then $G_n$ is equipped with the usual shortest-path distance, that we denote by $d^f$ as well.
We also use the following notation: $d:=\lfloor(\ln (\delta_n^3n))^{2/3}\rfloor$ when we prove the assertion for $1+\omega(n^{-1/3})= pn\leq \gamma$ and $d=\lfloor 3A\ln n\rfloor$ when we prove it for $1-n^{-1/3}\ln n\leq pn=1+o(1)$. In what follows, we assume that the assertions of Claim~\ref{cl:complex_components_sizes}, Claim~\ref{cl:kernel_complex}, Claim~\ref{cl:good_distinguish_large_p}, and Claim~\ref{cl:good_distinguish_small_p} hold deterministically in $G_n$.\\

{\bf 1.} Assume that some $v\notin D$ and $u\in D$ have $C(u)=C(v)$. We know that $v$ is $d$-close to a cycle $F$ of length at most $2d$. If $v\in V(F)$, then let $v'$ be the closest to $v$ vertex on $F$ that has degree more than 2 in the core. Otherwise, let $v'=v$. Let $P$ be the shortest path from $v'$ to $F$. Let us extend this path by a path $P'$ of length $10d$ beyond $v'$. Due to Claim~\ref{cl:complex_components_sizes} and Claim~\ref{cl:kernel_complex}, it has a subpath $w\ldots w'$ of length $5d$ consisting of vertices from $D$ only such that $d^f(w,v')\leq d$. We know that all elements of the vector $\mathbf{c}:=(C(w),\ldots,C(w'))$ are different. Then, due to our assumption, $u$ must have a vertex $z$ at distance at most $d^f(w,v)$ such that $z$ is the first vertex of a path $z\ldots z'$ with $C(w)=C(z)$, $\ldots,$ $C(w')=C(z')$. 

We now consider separately two cases: 1) $w=z$; 2) $w\neq z$. In the first case, we have that the distance from $u$ to the closest cycle (which is $F$ --- the same as for $v$) is at most 
$$ 
 d^f(w,u)+d^f(w,v')+d^f(v',F)\leq \text{length of }F+2(d^f(w,v')+d^f(v',F))\leq 6d.
$$
Let $P$ be the shortest path between $u$ and $v$. Due to Claim~\ref{cl:complex_components_sizes} and Claim~\ref{cl:kernel_complex}, there exists a path $u\tilde w\ldots \tilde w'$ of length $5d+1$ that does not meet $P$ and consists of vertices from $D$ only. Due to Claim~\ref{cl:good_distinguish_large_p} and Claim~\ref{cl:good_distinguish_small_p} all elements of the vector $\mathbf{\tilde c}:=(C(\tilde w),\ldots,C(\tilde w'))$ are different and no element of $\mathbf{\tilde c}$ equals to any element of $\mathbf{c}$. Moreover, by construction, $d^f(v,\tilde w)>d^f(u,\tilde w)$. Then, due to our assumption, $v$ must have a neighbor $\tilde z\neq\tilde w$ such that $\tilde z$ is the first vertex of a path $\tilde z\ldots \tilde z'$ with $C(\tilde w)=C(\tilde z)$, $\ldots,$ $C(\tilde w')=C(\tilde z')$. Note that $\tilde w\neq\tilde z,\ldots,\tilde w'\neq\tilde z'$ due to Claim~\ref{cl:complex_components_sizes} and Claim~\ref{cl:kernel_complex}. Since all vertices in $D$ are distinguished by $C(\cdot)$, we conclude that all vertices $\tilde z,\ldots,\tilde z'$ must be outside $D$. Due to Claim~\ref{cl:complex_components_sizes}, Claim~\ref{cl:kernel_complex}, and the definition of $D$, they constitute a (self-avoiding) path and are $d$-close to a cycle of length at most $2d$. Since the path has length $5d$, we get a contradiction with Claim~\ref{cl:complex_components_sizes} or Claim~\ref{cl:kernel_complex}.

We then assume $w\neq z$. It may only happen when $z\notin D$. Moreover, all $z,\ldots,z'$ are not in $D$. Indeed, otherwise, different paths $w\ldots w'$ and $z\ldots z'$ have common vertices. Then the path from $z$ to $F$ that goes through $w$ has length greater than $d$. However, due to Claim~\ref{cl:complex_components_sizes} and Claim~\ref{cl:kernel_complex}, there are no two different paths from $z$ to $F$, both of length at most $12d$ and, also, there is no other cycle $F'$ of length at most $2d$ such that a path from $z$ to $F'$ has at most $d$ vertices. This contradicts the fact that $z\notin D$.
 Thus, we again get a (self-avoiding) path consisting of vertices $z,\ldots,z'$ that are $d$-close to a cycle of length at most $2d$ --- contradiction with Claim~\ref{cl:complex_components_sizes} or Claim~\ref{cl:kernel_complex} again, since the path has length $5d$.

We conclude that every vertex $u\in D$ has $C(u)$ that does not equal to the color of any other vertex in the core.\\

{\bf 2.} It remains to prove that, for any two distinct $u,v\notin D$ that are not interchangeable, $C(u)\neq C(v)$. Fix two such vertices $u$ and $v$. We may assume that $T_u\cong T_v$ since otherwise $C(u)\neq C(v)$ due to Claim~\ref{lem:CGCH}. Let $F_u$ and $F_v$ be two cycles of length at most $2d$ that are closest to $u$ and $v$ respectively (both are at distance at most $d$ from the respective vertices). If $F_u\neq F_v$, then set $F:=F_u$. In this case, we let $u'=u$ when $u\notin V(F)$ and let $u'$ be the closest vertex of degree 3 in $F$ to $u$, otherwise. If $F_u=F_v=:F$, then, without loss of generality, we assume that either $u$ is not in $F$, or both $u,v$ are in $F$. Let $u'=u$ when $u\notin V(F)$ and let $u'$ be a vertex of $F$ that has degree at least 3 and such that $d^f(u,u')\neq d^f(v,u')$, otherwise. Note that such a vertex exists due to the definition of an interchangeable pair. Consider a path $P$ of length $10d$ that starts at $u'$ and does not meet $F$. Due to Claim~\ref{cl:complex_components_sizes} and Claim~\ref{cl:kernel_complex}, this path has a vertex $w$ from $D$ such that $d(w,u')\leq d$. If $F_u\neq F_v$, then 
$$
d^f(v,w)\geq d^f(F_u,F_v)-d^f(v,F_v)-d^f(w,F_u)> d^f(u,w)
$$
due to Claim~\ref{cl:complex_components_sizes} and Claim~\ref{cl:kernel_complex}. Finally, let $F_u=F_v$. Assume, in addition, $u\notin V(F)$. Then the only possibility for $C(u)$ to be equal to $C(v)$ is to have a path $P'$ between $v$ and $w$ of length $d^f(u,w)$. Let us extend $P'$ by a path of length $10d$ beyond $v$. Due to Claim~\ref{cl:complex_components_sizes} and Claim~\ref{cl:kernel_complex} this path has a vertex $w'$ from $D$ such that $d^f(w',v)\leq d$. But then $d^f(u,w')>d^f(v,w')$ implying $C(u)\neq C(v)$.  
 If $u,v\in V(F)$, then
$$
 d^f(v,w)=d^f(v,u')+d^f(u',w)\neq d^f(u,u')+d^f(u',w)=d^f(u,w).
$$
In either case, we get $d^f(v,w)\neq d^f(u,w)$ or $C(u)\neq C(v)$. Recalling that $w$ has a unique color, we immediately get $C(u)\neq C(v)$ always, completing the proof.

\subsection{Coloring the core in the critical phase}
\label{sc:main_critical}

Although in Section~\ref{sc:proof2_small} we color the entire union of complex components, here we show that the outcome of the CR on the core is also easy to analyze. 

Let 
$$p=1\pm O(n^{-1/3}),\quad G_n\sim G(n,p).
$$
We need the following result, see~\cite[Theorems~5.13,~5.21]{Janson_book}.

\begin{theorem}
The total excess of the union of complex components in $G_n$ is bounded in probability. All pendant paths have length $\Theta_P(n^{1/3}).$
\end{theorem}

From this we may conclude that whp all pendant paths in $G_n$ have different lengths. Indeed, first of all, the number of complex components is bounded in probability since each component increases the total excess. Moreover, the core of every complex component has size $O_P(n^{1/3})$. Let us fix a non-negative integer $t$, disjoint sets $V_1,\ldots,V_t\subset[n]$ of sizes $O(n^{1/3})$, and $t$ connected multigraphs $K'_1,\ldots,K'_t$ of bounded sizes with minimum degrees at least 3. Consider the event $\mathcal{E}(V_1,\ldots,V_t,K'_1,\ldots,K'_t)$ that 
\begin{itemize}
\item each $V_i$ spans the core of some complex component in $G_n$;
\item the core of the union of complex components is induced by $V_1\sqcup\ldots\sqcup V_t$;
\item the kernel of $V_i$ is isomorphic to $K'_i$ for every $i\in[t]$.
\end{itemize}
We get that
\begin{multline*}
 \sum_{t}\sum_{V_1,\ldots,V_t,K'_1,\ldots,K'_t}\mathbb{P}\left(\mathcal{E}(V_1,\ldots,V_t,K'_1,\ldots,K'_t)\right)=\\
 =\mathbb{P}\left(\bigcup_{t;\,V_1,\ldots,V_t,K'_1,\ldots,K'_t}\mathcal{E}(V_1,\ldots,V_t,K'_1,\ldots,K'_t)\right)=1-o(1).
\end{multline*}
On the other hand, let $\mathcal{E}'$ be the desired event saying that all pendant paths have different length. Clearly, subject to $\mathcal{E}(V_1,\ldots,V_t,K'_1,\ldots,K'_t)$ the probability that there exist two pendant paths with equal lengths is exactly the fraction of decompositions of each $|V_i|-|V(K'_i)|$ into $|E(K'_i)|$ non-negative summands such that at least two integers in the decomposition coincide. By the de Moivre--Laplace local limit theorem, this fraction is $O(n^{-1/6})$. Thus,
$$
 \mathbb{P}(\mathcal{E}')=
 \sum\mathbb{P}\left(\mathcal{E}'\mid\mathcal{E}(V_1,\ldots,V_t,K'_1,\ldots,K'_t)\right)\mathbb{P}\left(\mathcal{E}(V_1,\ldots,V_t,K'_1,\ldots,K'_t)\right)=1-o(1),
$$
as needed.

Now, note that, if two vertices from $\mathrm{C}_n$ that have degrees at least 3 in the core have different multisets of lengths of pendant paths that are adjacent to them, then CR run on $\mathrm{C}_n$ color them differently due to Claim~\ref{cl:locally_trees_CR_2}. This immediately implies that whp the only possibility for two vertices $u\neq v$ from $\mathrm{C}_n$ of degrees at least 3 to be in the same color class is to belong to the same 1-complex component such that all its three pendant paths join $u$ and $v$. From Lemma~\ref{lem:UvsC}, it then follows that for any not interchangeable pair of vertices $x,y$ that have degree 2 in $\mathrm{C}_n$, $C(u)=C(v)$ only when $x$ and $y$ belong to the same pendant path of a 1-complex component as above and $d(x,u)=d(y,v)$, where $u,v$ are the vertices with degree 3 of the core of this component. We get the following lemma that complements our Main Lemma.

\begin{lemma}
Let $p=1\pm O(n^{-1/3})$, $G_n\sim G(n,p)$, and let $\mathrm{C}_n$ be the core of the union of complex components of $G_n$. Let $C$ be the coloring produced by CR run on $\mathrm{C}_n$. Whp, for any two vertices $u\neq v$ of $\mathrm{C}_n$, if $C(u)=C(v)$, then one of the following two options holds:
\begin{itemize}
\item the pair $\{u,v\}$ is interchangeable;
\item there exists a 1-complex component in $\mathrm{C}_n$ with two vertices $x,y$ of degree 3 and three pendant paths $P_1,P_2,P_3$ between $x$ and $y$ such that $u,v\in P_1$ and $d_{P_1}(u,x)=d_{P_1}(v,y)$ (it may also happen that $u=x$ and $v=y$).
\end{itemize}
\label{lm:critical_CR_core}
\end{lemma}

\section{Proofs of Theorem~\ref{thm:ColRef}~and~Theorem~\ref{th:symmetries}}
\label{sc:proofs_th_Intro}

In this section, we prove Theorems~\ref{thm:ColRef}~and~\ref{th:symmetries}. In the proofs, we call the union of a pair of transposable paths a {\it degenerate cycle}. Let $\Sigma_1,\Sigma_2$ be the set of all pendant and degenerate cycles in $G(n,p)$, respectively. Set $\Sigma:=\Sigma_1\sqcup\Sigma_2$.

\subsection{Proof of Theorem~\ref{thm:ColRef}}
\label{sc:th1_proof}

Due to Main Lemma (Lemma~\ref{th:giant_main}), it remains to prove the following.

\begin{lemma}
Let $\gamma>1$ be a constant and $pn\leq\gamma$. Then $|\Sigma|=O_P(1)$ and whp there are no two cycles in $\Sigma$ that share a vertex.
\label{lm:loops_multiple}
\end{lemma}

Indeed,  Lemma~\ref{lm:loops_multiple} implies that all color classes have cardinalities at most 2 whp and their union is stochastically bounded. The fact that each color class is an orbit immediately follows from the definition of an interchangeable pair and the fact that, if an automorphism of $\mathrm{H}_n$ maps $x$ to $y$, then $C(x)=C(y)$. 
 We prove Lemma~\ref{lm:loops_multiple} in Section~\ref{sc:lm_degenerate_cycles_proof}.

\subsection{Proof of Theorem~\ref{th:symmetries}} 

Let $\gamma>1$ and $G_n\sim G(n,p)$. Let $\mathrm{C}_n$ be the core of the largest component of $G_n$. If an automorphism of $\mathrm{C}_n$ maps a vertex $x$ to a vertex $y$, then $C^{\mathrm{C}_n}(x)=C^{\mathrm{C}_n}(y)$ due to the definition of CR. Therefore, if $1+\omega(n^{-1/3})=pn\leq\gamma$, then, due to Main Lemma (Lemma~\ref{th:giant_main}), the only possible automorphisms of $\mathrm{C}_n$ are permutations of interchangeable vertices, that belong to cycles from $\Sigma$. If $np=1\pm O(n^{-1/3})$, then, due to Lemma~\ref{lm:critical_CR_core}, the only possible automorphisms of $\mathrm{C}_n$ are permutations of interchangeable vertices, that belong to cycles from $\Sigma$, and transpositions of 1-complex components described in this lemma --- we will denote the set of such components by $\Sigma_3$ in what follows.

Due to Lemma~\ref{lm:loops_multiple},  $|\Sigma|=O_P(1)$. Moreover, if $1+\omega(n^{-1/3})=pn\leq\gamma$, then  $\aut(\mathrm{C}_n)\cong\prod_{F\in\Sigma}S_2$ whp; if 
$np=1\pm O(n^{-1/3})$, then $\aut(\mathrm{C}_n)\cong\prod_{F\in\Sigma\sqcup\Sigma_3}S_2$ due to Lemma~\ref{lm:critical_CR_core}.
It remains to prove the following.
\begin{lemma}
Let $\delta>0$ and $\gamma>1+\delta$ be constants. Let $pn\leq\gamma$. 
\begin{enumerate}
\item If $pn\geq 1+\delta$, then
both $\Sigma_1,\Sigma_2$ are non-empty with non-vanishing probability.
\item If $pn=1+\omega(n^{-1/3})$ and $pn=1+o(1)$, then whp $\Sigma_2=\varnothing$, while $\Sigma_1$ is non-empty with non-vanishing probability.
\item If $pn=1\pm O(n^{-1/3})$, then whp $\Sigma_2=\varnothing$, while $\Sigma_1,\Sigma_3$  are non-empty with non-vanishing probability.
\end{enumerate}
\label{lm:loops_exist}
\end{lemma}

We prove Lemma~\ref{lm:loops_exist} in Section~\ref{sc:lm_degenerate_cycles_proof}. It completes the proof of Theorem~\ref{th:symmetries}.

\subsection{Proofs of Lemma~\ref{lm:loops_multiple} and Lemma~\ref{lm:loops_exist}}
\label{sc:lm_degenerate_cycles_proof}

In this section, we prove Lemma~\ref{lm:loops_multiple} and Lemma~\ref{lm:loops_exist} simultaneously. We divide the proof into three parts: (1) $1+\delta\leq np\leq \gamma$ (Section~\ref{sc:asym_strictly_super_critical}); (2) $np=1\pm O(n^{-1/3})$ (Section~\ref{sc:asym_critical}); (3)~$np=1+\omega(n^{-1/3})$ and $np=1+o(1)$ (Section~\ref{sc:asym_super_critical}).

\subsubsection{Strictly supercritical regime: proof of Lemma~\ref{lm:loops_multiple} and the first part of Lemma~\ref{lm:loops_exist}}
\label{sc:asym_strictly_super_critical}

Here we assume that $1+\delta\leq np\leq\gamma$.

\paragraph{Proof of Lemma~\ref{lm:loops_multiple}.} The fact that whp there are no two cycles in $\Sigma$ that share a vertex follows immediately from Claim~\ref{cl:kernel_complex}.\\

Let us prove that $|\Sigma|=O_P(1)$. It is well-known that whp the trees attached to the core of the giant in strictly supercritical random graph have size $O(\ln n)$ (see, e.g.,~\cite[Lemma 10]{BohmanFLPSSV07}). In particular, it immediately follows from Theorem~\ref{th:contiguous_super_critical} and the fact that the total progeny $Y$ of a subcritical Galton--Watson process with expected number of offsprings $\mu<1$ admits exponential decay of tail probabilities: there exists a constant $c$ such that, for every $t$, $\mathbb{P}(Y>t)\leq e^{-ct}$. The latter fact is fairly standard. It can be proven directly using Cram\'{e}r--Chernoff method and the equation for the generating function of the total progeny~\cite{Feller}. Alternatively, it immediately follows from the fact that $Y$ coincides in distribution with the first time when the \L ukasiewicz random walk hits $-1$, see, e.g.,~\cite{Kort}.

Let us define the event $\mathcal{E}_n$ that the trees attached to the core of the giant have size at most $\ln^2 n$. We have that $\mathcal{E}_n$ holds whp in $G_n$. Let $U\subset[n]$ be a fixed set  of size $n':=n-\lceil\sqrt{n}\rceil$. Fix $u\in U$. The probability that the connected component in $G_n[U]$ that contains $u$ is a tree of size less than $\ln^2n$ is at most the extinction probability $\rho$ of the Galton--Watson process with offspring distribution $\mathrm{Bin}(n',p)$. Since the generating function of this distribution equals
$$
 g(x)=(px+1-p)^{n'}=e^{n'p(1+O(p))(x-1)},
$$
we get that 
$$
\rho=e^{n'p(\rho-1)(1+O(1/n))}.
$$ 
There exists a constant $\delta'>0$ such that $\rho<\frac{1}{np}-\delta'$. Indeed, $(g(x)-x)|_{x=0}>0$ while $g(1/np)-1/np=e^{1+o(1)-np}-1/np<0$.

Let $A$ be a sufficiently large constant, $k=\lfloor A\ln n\rfloor$ and let $X_{k}$ be the number of $k$-paths so that all their vertices do not belong to the kernel. Then
$$
\mathbb{E}\left(X_{k}\cdot\1_{\mathcal{E}_n}\right)\leq n^kp^{k-1}\rho^{k}\leq\frac{1}{p}(1-\delta'(1+\delta))^k=o(1).
$$
It immediately implies that whp there are no $\mathrm{C}_n$-degenerate cycles of length more that $3A\ln n$.

Fix $v\in[n]$, and let $\xi_v$ be the number of cycles of lengths at most $3A\ln n$ that contain $v$ such that all their vertices, but $v$ and the opposite to $v$ vertex on the cycle, have degree~2 in $\mathrm{C}_n$. From the above, we get
$$
 \mathbb{E}(\xi_v\cdot\1_{\mathcal{E}_n})\leq\sum_{k=3}^{\lfloor 3A\ln n\rfloor}n^{k-1}p^k\rho^{k-2}<\frac{1}{n}(np)^2\sum_{k=3}^{\infty}(1-\delta')^{k-2}<\frac{\gamma^2}{\delta' n}.
$$
This completes the proof due to Markov's inequality since, letting $\xi$ be the number of $\mathrm{C}_n$-degenerate cycles, we get that, for every constant $B>0$,
\begin{align*}
\mathbb{P}(\xi>B) & \leq \mathbb{P}(\{\xi>B\}\wedge\mathcal{E}_n)+\mathbb{P}(\neg\mathcal{E}_n)=
\mathbb{P}(\{\xi\cdot\1_{\mathcal{E}_n}>B\}\wedge\mathcal{E}_n)+o(1)\\
&\leq
 \mathbb{P}(\xi\cdot\1_{\mathcal{E}_n}>B)+o(1)\leq\frac{1}{B}\mathbb{E}(\xi\cdot\1_{\mathcal{E}_n})+o(1)\\
 &=\frac{1}{B}\sum_{v=1}^n\mathbb{E}(\xi_v\cdot\1_{\mathcal{E}_n})+o(1)\leq \frac{2\gamma^2}{B\delta'}.
\end{align*}

\begin{remark}
This proof generalizes to smaller $p$ --- say, $1+1/\ln\ln n\leq np=:1+\delta_n\leq 1+\delta$. It is sufficient to note that whp 1) for every $v\in V(\mathrm{C}_n)$, the tree $T_v$ has less than $\ln^2 n$ vertices; 2) for a fixed set $U$ of size at least $n-\sqrt{n}$ and a vertex $u\in U$, the probability that the connected component in $G_n[U]$ that contains $u$ is a tree of size less than $\ln^2n$ is at most $(1-\delta_n/2)/(np)$; 3) $\mathrm{C}_n$ does not contain a path of length more than $\ln^2n$ with all its vertices outside of the kernel.
\end{remark}

\paragraph{Proof of the first part of Lemma~\ref{lm:loops_exist}.} It remains to prove that, for some constant $\varepsilon>0$, 
$$
\mathbb{P}(\Sigma_1\neq\varnothing,\Sigma_2\neq\varnothing)>\varepsilon.
$$
Let us fix seven distinct vertices $u,\tilde u_1,\tilde u_2$, $v_1,v_2,\tilde v_1,\tilde v_2$ and consider the following event $\mathcal{E}_{u,\tilde u_1,\tilde u_2; v_1,v_2,\tilde v_1,\tilde v_2}$: 
\begin{itemize}
\item $u\tilde u_1\tilde u_2$ is a triangle from $\mathrm{C}_n$ with both $\tilde u_1,\tilde u_2$ non-adjacent to any vertex outside of the triangle, 
\item $v_1\tilde v_1v_2\tilde v_2$ is a 4-cycle from $\mathrm{C}_n$ with both $\tilde v_1,\tilde v_2$ non-adjacent to any vertex outside of the 4-cycle.
\end{itemize} 
Let $X$ be the number of events $\mathcal{E}_{u,\tilde u_1,\tilde u_2; v_1,v_2,\tilde v_1,\tilde v_2}$ that occur in $G_n$.

Consider the event $\mathcal{E}'$ that (1) $G_n$ has a component of size at least $\delta'n$, where $\delta'$ is a small enough constant; (2) all trees sprouting from the core have size $o(\ln n)$. We have that $\mathbb{P}(\mathcal{E}')=1-o(1)$ (for example, due to Theorem~\ref{th:contiguous_super_critical}). Let $q_n(N)$ be the probability that a fixed vertex belongs to a component of size at least $\sqrt{n}$ in a fixed set of size $N$. We know that, for any $N=n(1-o(1))$, $q_n(N)\sim q_n(n)=:q$ is bounded away both from 0 and 1 (see, e.g.,~\cite[Theorems 5.1, 5.4]{Janson_book}).  By exploring the giant from $u$ in-width, we may stop the exploration as soon as its size reaches $\lceil\sqrt{n}\rceil$ and switch to the vertex $v_1$. Thus, we may treat the respective events for $u,v_1,v_2$ independently. 
 Note that, as soon as $\mathcal{E}'$ holds, we have that, if $v_1\sim \tilde v_1\sim v_2$ and both $v_1,v_2$ belong to components of size at least $\sqrt{n}$ in $G_n\setminus \tilde v_1$, then $v_1,v_2$ belong to the core. We have that 
$$
\mathbb{P}(\mathcal{E}_{u,\tilde u_1,\tilde u_2; v_1,v_2,\tilde v_1,\tilde v_2}\cap\mathcal{E}')=(np)^7(1-p)^{4n-O(1)}(q+o(1))^3=(1+o(1))c^7e^{-4c}q^3n^{-7}.
$$
We get that 
$$
\mathbb{E}(X\cdot \1_{\mathcal{E}'})=(1+o(1))\frac{n^7}{8}(np)^7e^{-4c}q^3n^{-7}=\frac{1}{8}(np)^7e^{-4c}q^3+o(1)=\Theta(1).
$$
Now, fix  tuples $\mathbf{a}^i=(u^i,\tilde u^i_1,\tilde u^i_2; v^i_1,v^i_2,\tilde v^i_1,\tilde v^i_2)$, $i\in\{1,2\}$. Note that any collection of such tuples on $v<14$ vertices contributes to $\mathbb{E}(X\cdot\1_{\mathcal{E}'})^{(m)}$ the probability 
$$
\mathbb{P}\left(\wedge_{i=1}^2\mathcal{E}_{u,\tilde u^i_1,\tilde u^i_2; v^i_1,v^i_2,\tilde v^i_1,\tilde v^i_2}\right)=O(n^{-v-1})
$$ 
since the total number of required edges is at least $v+1$, while the requirement on the absence of edges contributes a $\Theta(1)$-factor. Thus,
\begin{multline*}
\mathbb{E}(X(X-1)\cdot\1_{\mathcal{E}'})  =\left((1+o(1))\frac{n^7}{8}\right)^2\left(p^7(1-p)^{4n-O(1)}(q+o(1))^3\right)^2+\\
+O\left(\sum_{v=8}^{7m-1}n^v n^{-v-1}\right)=\left(\frac{1}{8}(np)^7e^{-4c}q^3\right)^2+o(1).
\end{multline*}
This completes the proof due to the Paley--Zygmund inequality.

\subsubsection{Critical phase: proof of Lemma~\ref{lm:loops_multiple} and the third part of Lemma~\ref{lm:loops_exist}}
\label{sc:asym_critical}

Let $np=1+O(n^{-1/3})$. We set $\delta_n:=np-1$.
First of all, the fact that, for some constant $\varepsilon>0$, 
$$
\mathbb{P}(\Sigma_1\neq\varnothing,\Sigma_3\neq\varnothing)>\varepsilon
$$
easily follows from~\cite[Theorem 5.20]{Janson_book} that, in particular, states that with probability greater than $\varepsilon$, $G(n,p)$ contains a 1-complex component. It remains to note that there are two types of the core of 1-complex components: 1) disjoint union of two cycles joined by a path; 2) a cycle with a path joining two its vertices. Letting $f_i(k)$ be the number of such (labeled) graphs of type $i$, $i\in\{1,2\}$, on $[k]$ we get that $f_1=\Theta(f_2)$. Since any graph of the first type contributes to $\Sigma_1$ and any graph of the second type contributes to $\Sigma_3$, the fact that $\mathbb{P}(\Sigma_1\neq\varnothing,\Sigma_3\neq\varnothing)$ is bounded away from 0, follows.

It remains to prove Lemma~\ref{lm:loops_multiple} and that whp $\Sigma_2=\varnothing.$ We need the following assertion, see~\cite[Theorem 4]{LPW}.
\begin{claim}
Every connected component of $\mathrm{C}_n$ has diameter $O_P(n^{1/3})$. 
\label{cl:critical_diameter}
\end{claim}

Let $a_n\to\infty$ as $n\to\infty.$
Below, we prove that whp $\Sigma_2=\varnothing$, $|\Sigma_1|\leq e^{O(a_n)}$\footnote{It is easy to show an equivalent definition of stochastic boundedness: $\xi_n=O_P(1)$ if and only if, for every increasing and unbounded sequence $a_n$, $\mathbb{P}(\xi_n<a_n)\to 1$ as $n\to\infty$.}, and there are no overlapping cycles in $\Sigma_1$.
From Claim~\ref{cl:critical_diameter}, it follows that, whp, 
\begin{itemize}
\item if $F\in\Sigma_2$ is a cycle with vertices $v_1,v_2$ of degree at least  3, then $F$ has length at most $a_n n^{1/3}$, and there is a path between $v_1,v_2$ with all inner vertices outside of $F$ of length at most $a_n n^{1/3}$;
\item if $F\in\Sigma_1$ is a cycle with a vertex $v$ of degree at least  3, then $F$ has length at most $a_n n^{1/3}$, there exists another cycle of length at most $a_n n^{1/3}$, that is joined with $F$ by a path of length at most $a_n n^{1/3}$.
\end{itemize}
Then
\begin{align*}
 \mathbb{E}|\Sigma_2|
 &\leq \sum_{k_1=4}^{\lfloor a_n n^{1/3}\rfloor}\sum_{k_2=1}^{\lfloor a_n n^{1/3}\rfloor}n^{k_1+k_2}p^{k_1+k_2+1}=p\sum_{k_1=4}^{\lfloor a_n n^{1/3}\rfloor}\sum_{k_2=1}^{\lfloor a_n n^{1/3}\rfloor}(1+\delta_n)^{k_1+k_2}\\
 &\leq\frac{(a_n n^{1/3})^2}{n}(1+\delta_n)^{2a_n n^{1/3}}
 \leq\frac{(1+o(1))a_n^2}{n^{1/3}} e^{O(a_n)}=o\left(1\right)
\end{align*}
and
\begin{align*}
 \mathbb{E}|\Sigma_1|
 &\leq \sum_{k_1=3}^{\lfloor a_n n^{1/3}\rfloor}\sum_{k_2=1}^{\lfloor a_n n^{1/3}\rfloor}\sum_{k_3=0}^{\lfloor a_n n^{1/3}\rfloor}n^{k_1+k_2+k_3}p^{k_1+k_2+k_3+1}\\
 &=p\sum_{k_1=4}^{\lfloor a_n n^{1/3}\rfloor}\sum_{k_2=1}^{\lfloor a_n n^{1/3}\rfloor}\sum_{k_3=0}^{\lfloor a_n n^{1/3}\rfloor}(1+\delta_n)^{k_1+k_2+k_3}\\
 &\leq\frac{(a_n n^{1/3})^3}{n}(1+\delta_n)^{3a_n n^{1/3}}
 \leq
 (1+o(1))a_n^3 e^{O(a_n)}=e^{O(a_n)}.
\end{align*}
Finally, let $X$ be the number of unions of two overlapping cycles from $\Sigma_1$. Then
$$
 \mathbb{E}X
 \leq \sum_{k_1=3}^{\lfloor a_n n^{1/3}\rfloor}\sum_{k_2=3}^{\lfloor a_n n^{1/3}\rfloor}n^{k_1+k_2}p^{k_1+k_2+1}\leq\frac{(a_n n^{1/3})^2}{n}(1+\delta_n)^{2a_n n^{1/3}}
 =o\left(1\right),
$$
completing the proof.

\subsubsection{Supercritical phase: proof of Lemma~\ref{lm:loops_multiple} and the second part of Lemma~\ref{lm:loops_exist}}
\label{sc:asym_super_critical}

It remains to consider the case $np=1+\omega(n^{-1/3})$ and, simultaneously, $np=1+o(1)$. Let us first show that, for some constant $\varepsilon>0$,
\begin{equation}
 \mathbb{P}(\Sigma_1\neq\varnothing)>\varepsilon.
\label{eq:critical_Sigma1_non-trivial}
\end{equation}

Let us fix some positive numbers $a<b$ and let $X$ be the number of cycles from $\Sigma_1$ of length $k\in[a/\delta_n,b/\delta_n]$, where $\delta_n=np-1$. Let $\mathcal{E}_n$ be the event that
\begin{itemize}
\item the giant component in $G_n$ has size $\Theta(n\delta_n)$;
\item the trees attached to the core of the giant have size $O(1/\delta_n).$
\end{itemize}
Then $\mathcal{E}_n$ holds whp and implies that the deletion of any 2-connected block from $G_n$ of size $O(1/\delta_n^2)$ keeps the size of the giant asymptotically unchanged. Let $U\subset[n]$ be a fixed set of size $n-O(1/\delta_n^2)$, $u\in U$. Then the probability that the connected component in $G_n[U]$ that contains $u$ is a tree of size $O(1/\delta_n)$ equals $1-(2+o(1))\delta_n$, which is the extinction probability of the Galton--Watson process with offspring distribution $\mathrm{Bin}(|U|,p)$. We conclude that
\begin{align*}
 \mathbb{E}X & =\Theta(n\delta_n)\times \sum_{k\in[a/\delta_n,b/\delta_n]}n^kp^{k+1}(1-(2+o(1))\delta_n)^k\\
 &=\Theta(\delta_n)\times\sum_{k\in[a/\delta_n,b/\delta_n]}(1-(1+o(1))\delta_n)^k=\Theta(1).
\end{align*}
In the same way, we estimate the second moment. Clearly, two cycles from $\Sigma_1$ may only overlap in the single vertex that has degree more than 2 in $\mathrm{C}_n$. Thus,
$$
 \mathbb{E}X(X-1)=\left(\Theta(n\delta_n)\times \sum_{k\in[a/\delta_n,b/\delta_n]}n^kp^{k+1}(1-(2+o(1))\delta_n)^k\right)^2=\Theta((\mathbb{E}X)^2).
$$
Thus, $\mathbb{E}X^2=\mathbb{E}X(X-1)+\mathbb{E}X=\Theta((\mathbb{E}X)^2)$. Paley--Zygmund inequality completes the proof of~\eqref{eq:critical_Sigma1_non-trivial}.

It remains to prove that $\Sigma_1=O_P(1)$, whp there are no two cycles in $\Sigma_1$ that share a vertex and $\Sigma_2=\varnothing.$ To show this, we will use the contiguous model from~\cite[Theorem 2]{DKLP:anatomy} that we recalled in Section~\ref{sc:complex_structure}. 

In order to analyze the contiguous model, we will use the Bollobas's configuration model~\cite{Bol_configuration} as well as switchings, a very useful tool that allows to study random regular graphs and uniformly random graphs with a specified degree sequence (see, e.g.,~\cite{Hasheminezhad_McKay} for a general overview of the method). 

Assume that we have already exposed the values of $\eta_i$ and the event $\mathcal{B}$ occurs. Notice that whp $\sum\eta_i\1_{\eta_i\geq 3}=\Theta(N)$, $\sum\eta_i^2\1_{\eta_i\geq 3}=\Theta(N)$, and, for every $i\in[n]$, $\eta_i<\ln n$. So, in what follows, we assume that
the degree sequence of the random multigraph on $[N]$ is fixed and the above properties hold deterministically. In particular, assuming that the vertex $i\in[N]$ has degree $d_i$, we get that the number of edges equals $E:=\frac{1}{2}\sum_{i=1}^N d_i=\Theta(N)$, $\sum_{i=1}^N d_i^2=\Theta(N)$, and all $d_i$ are at most $\ln n$. 

Let us now recall the {\it configuration model}. Associate each vertex $i\in[N]$ with a ``block'' of $d_i$ distinct points (also referred to as ``half-edges'') and consider a uniform perfect matching on these points. Let $\mathcal{M}$ be the set of all matchings and $\Pi$ be the set of all multigraphs. Clearly, there exists an injection $f:\mathcal{M}\to\Pi$ such that every simple graph corresponds to exactly $\prod_{i=1}^Nd_i!$ matchings:
if there is an edge of the matching $M\in\mathcal{M}$ between the $i$-th block and the $j$-th block, then $i$ and $j$ are adjacent in $f(M)$. In particular, every matching that does not have repeated edges between the blocks and does not have edges within a block corresponds to a simple graph. For $M\in\mathcal{M}$, let $w(M):=|f^{-1}(f(M))|$ be the {\it weight} of the matching $M$. Clearly, $|\Pi|=\sum_{M\in\mathcal{M}}\frac{1}{w(M)}$.\\

We now can estimate the probability that a given vertex belongs to a loop using switchings. Fix $v\in[N]$. Let $\Pi_v^1\subset\Pi$ be the set of all multigraphs that have a single loop at $v$ and $\Pi_v^0\subset\Pi$ be the set of all multigraphs that have no loops at $v$. In the same way, let $\mathcal{M}_v^1\subset\mathcal{M}$ be the set of all matchings that have a single edge in the block that corresponds to $v$ and let $\mathcal{M}_v^0\subset\mathcal{M}$ be the set of all matchings that have no edges in the block that corresponds to $v$.

Consider the bipartite graph $\mathcal{H}_v$ with parts $\mathcal{M}_v^1$ and $\mathcal{M}_v^0$ with an edge between $M\in\mathcal{M}_v^1$ and $M'\in\mathcal{M}_v^0$ if and only if there are two nodes $x,y$ in the block that corresponds to the vertex $v$ and two nodes $x',y'$ outside of this block so that $x\sim y$ and $x'\sim y'$ in $M$, while $x\sim x'$ and $y\sim y'$ in $M'$. Note that every $M\in \mathcal{M}_v^1$ has degree $E-(d_v-1)$, while every $M'\in\mathcal{M}_v^0$ has degree ${d_v\choose 2}$. We also note that $w(M)=\frac{1}{2}w(M')$. Thus,
\begin{align}
 |\Pi^1_v| &=\sum_{M'\in\mathcal{M}^1_v}\frac{1}{w(M')}=
 \sum_{(M,M')\in E(\mathcal{H}_v)}\frac{1}{{d_v\choose 2}w(M')}\notag\\
 &=
  \sum_{(M,M')\in E(\mathcal{H}_v)}\frac{1}{2{d_v\choose 2}w(M)}=
  \sum_{M\in\mathcal{M}^0_v}\frac{E-(d_v-1)}{2{d_v\choose 2}w(M)}=
  \frac{E-(d_v-1)}{2{d_v\choose 2}}|\Pi^0_v|.
 \label{eq:switchings}
\end{align}
We immediately get that $v$ has a single loop with probability 
$$
\frac{|\Pi^0_v|}{|\Pi|}\leq\frac{|\Pi^0_v|}{|\Pi^1_v|}=\frac{2{d_v\choose 2}}{E-(d_v-1)}<\frac{d_v^2}{E-(d_v-1)}.
$$
In exactly the same way, we may show that the probability that the vertex $v$ has $t$ loops is asymptotically less than the probability that $v$ has $t-1$ loops. Thus, $v$ belongs to a loop with probability at most $(1+o(1))\frac{d_v^2}{E}$ and $v$ belongs to at least two loops with probability $o\left(\frac{d_v^2}{E}\right)$, uniformly over $v$. Therefore, the expected number of vertices that belong to a loop is at most $(1+o(1))\sum\frac{\eta_1^2+\ldots+\eta_n^2}{E}=\Theta(1)$ and the expected number of vertices that belong to at least two loops is $o(1)$. In particular, due to Theorem~\ref{th:contiguous_critical}, we get that whp there are no two cycles in $\Sigma_1$ that share a vertex. Moreover, for every increasing sequence $\{a_n,n\in\mathbb{N}\}$ such that $\lim_{n\to\infty}a_n=\infty$, whp the number of loops is less than $a_n$. Due to Theorem~\eqref{th:contiguous_critical}, the same property holds in the kernel $K(G_n)$, thus $|\Sigma_1|=O_P(1)$.\\

It remains prove that whp $\Sigma_2=\varnothing$. For this, we show that the expected number of multiple edges is $\Theta(1)$ as well. We follow exactly the same strategy as for loops with the following changes in the proof: now, $\mathcal{M}_{v,u}^1$ is the set of all matchings that have exactly two edges between the blocks that corresponds to $v$ and $u$ and $\mathcal{M}_{v,u}^1\subset\mathcal{M}$ is the set of all matchings that have at most one edge between these blocks. We also consider the bipartite graph $\mathcal{H}_{v,u}$ with parts $\mathcal{M}_{v,u}^1$ and $\mathcal{M}_{v,u}^0$ with an edge between $M$ and $M'$ if and only if there are two nodes $x,y$ in the block that corresponds to the vertex $v$, two nodes $x',y'$ in the block that corresponds to the vertex $u$, and vertices $x'',x''',y'',y'''$ that do not belong to the blocks that corresponds to $v,u$, so that $x\sim x'$, $x''\sim x'''$, $y\sim y'$, and $y''\sim y'''$ in $M$, while $x\sim x''$, $x'\sim x'''$, $y\sim y''$, and $y'\sim y'''$ in $M'$. Note that every $M\in \mathcal{M}_{v,u}^1$ has degree at least $(E-(d_v-d_u))^2$, while every $M'\in\mathcal{M}_{v,u}^0$ has degree at most $2{d_v\choose 2}{d_u\choose 2}$. We also note that $w(M)\geq\frac{1}{16}w(M')$. In the same way as in~\eqref{eq:switchings}, we get
$$
 |\Pi^1_{v,u}| \geq
  \frac{(E-(d_v-d_u))^2}{32{d_v\choose 2}{d_u\choose 2}}|\Pi^0_{v,u}|,
$$
implying that the probability that there are exactly two edges between $v$ and $u$ is at most has a single loop with probability at most $(1+o(1))\frac{8d_v^4d_u^2}{E^2}$. Again, we can similarly argue that the probability that there are $t$ edges between $v$ and $u$ is asymptotically less than the probability that $v$ sends $t-1$ edges to $u$. Therefore, the expected number of multiple edges is at most 
$$
(1+o(1))\frac{8\sum_{u\neq v}d_u^2d_v^2}{E^2}\leq
(1+o(1))\frac{4(\sum_u d_u^2)^2}{E^2}=\Theta(1)
$$
as well.

Let $a_n\to\infty$ as $n\to\infty$ slowly enough so that, for every positive integer $k$, $\mathbb{P}(\zeta=k)=o(1/a_n)$ for $\zeta\sim\mathrm{Geom}(1-\mu)$. Clearly, it is enough to take $a_n=o(1/\delta_n)$. For every pair of repeated edges, let us expose the value of the respective geometric variable for one of the two edges. Then, the probability that the value of the second edge is exactly the same is $o(1/a_n)$. Since whp the number of multiple edges is $o(a_n)$, we get that whp there are no repeated edges that are transformed to paths of the same length. Due to Theorem~\ref{th:contiguous_critical}, whp $\Sigma_2=\varnothing$.

\section*{Acknowledgements}

The authors would like to thank Michael Krivelevich for helpful discussions.

\bibliographystyle{abbrv}
\bibliography{evolution}

\appendix

\section{Universal algorithm}
\label{sc:appendix}

\newcommand{\OV}[1]{{\color{purple}{#1}}}

We here show that there exists a polynomial time algorithm that, for any function $p=p(n)$ with values in $[0,1]$,
whp produces a canonical labeling of~$G(n,p)$.

Note that if $1/n\ll p\leq 1/2$, then the core of the giant component of $G(n,p)$ coincides with the core of
the entire graph because the expected number of unicyclic components in $G(n,p)$ is at most 
\begin{align*}
 \sum_{k=3}^n {n\choose k} k^{k-2} k^2 p^k (1-p)^{{k\choose 2}-k+k(n-k)}
 &\leq
 \sum_{k=3}^n (enp(1-p)^{(n-3)/2})^k<\sum_{k=1}^{\infty}(e^{3-np/2}np)^k\\
 &=\frac{e^{3+\ln(np)-np/2}}{1-e^{3+\ln(np)-np/2}}=e^{-np(1/2-o(1))}=o(1).
\end{align*}
Due to the classical linear-time algorithms for canonical labeling of trees,
this observation reduces canonical labeling of $G(n,p)$ with $1/n\ll p\leq 1/2$ to
canonical labeling of its core. Note also that the core in this regime is whp asymmetric; see~\cite{LM}.

Linial and Mosheiff~\cite{LM} suggested an algorithm ${\sf A}_1$ that, for any $p$ with $\frac{1}{n}\ll p(n)< n^{-2/3}$,
whp labels canonically $G(n,p)$ in time $O(n^4)$ by distinguishing between all vertices of the core.
In~\cite{CzP}, it was proved that, if $\frac{\ln^4 n}{n}\leq p\leq\frac{1}{2}$, then CR whp distinguishes between
all vertices of the entire $G(n,p)$. Note that, in this regime, whp $G(n,p)$ has minimum degree at least 2, i.e. its core coincides with the entire graph.
 Finally, in the present paper we designed an algorithm ${\sf A}_2$ that, for any $p=O(1/n)$, whp labels canonically $G(n,p)$
in time $O(n\ln n)$. This algorithm, which is described in Section~\ref{sc:from_CR_to_CL},
succeeds on all graphs with $O(n)$ edges that satisfy the conclusion of Theorem~\ref{thm:ColRef}.

Now, consider the following algorithm ${\sf A}$:
\begin{enumerate}
\item Run CR. If it colors differently all vertices, then halt and output the canonical labeling produced by~CR.
\item If the algorithm does not halt in Step 1, then run ${\sf A}_1$. If it succeeds (i.e., colors differently all vertices
  in the core of the input graph), then halt and output the labelling produced by~${\sf A}_1$.
\item If the algorithm does not halt in Steps 1 and 2, then run ${\sf A}_2$ and output the labeling it produces
  (or give up if ${\sf A}_2$ fails). 
\end{enumerate}

Let us show that the algorithm ${\sf A}$ succeeds whp for any $p$ with $p(n)\leq 1/2$.
Assume, to the contrary, that there exist a constant $\varepsilon>0$ and a sequence $(n_k)_{k\in\mathbb{N}}$ such that 
$$
\mathbb{P}({\sf A}\text{ fails on $G(n_k,p(n_k))$})>\varepsilon
$$
for all $k$. If there is a subsequence $(n_{k_i})_{i\in\mathbb{N}}$ and a constant $C>0$ such that $p(n_{k_i})<C/n_{k_i}$ for all $i$,
then we get a contradiction with the performance of the algorithm ${\sf A}_2$. Therefore, $p(n_k)\gg\frac{1}{n_k}$.
If there is a subsequence $(n_{k_i})_{i\in\mathbb{N}}$ such that $p(n_{k_i})<n^{-2/3}_{k_i}$ for all $i$, then we get a contradiction
with the performance of the algorithm ${\sf A}_1$. It follows that $p(n_k)\geq n^{-2/3}_k$ for all $k$. This, however, contradicts
the result of~\cite{CzP} that CR in this regime produces a discrete coloring of $G(n,p)$ whp.

In order to obtain canonical labeling, whp, for all $p$ with $p(n)\in[0,1]$, we run the algorithm ${\sf A}$
on input $G$ and if it fails, then we run ${\sf A}$ once again on the complement of~$G$.
\end{document}